\documentclass[a4paper,11pt,reqno,noindent]{amsart}
\usepackage[centertags]{amsmath}
\usepackage{mathtools} 
\usepackage{amsfonts,amssymb,amsthm,dsfont,cases,amscd,esint,enumerate}
\usepackage[T1]{fontenc}
\usepackage[english,frenchb]{babel}
\usepackage[applemac]{inputenc}
\usepackage{newlfont,cancel}
\usepackage{color}
\usepackage[body={15cm,21.5cm},centering]{geometry} 
\usepackage{fancyhdr}
\usepackage{enumerate}
\usepackage{tikz}
\usetikzlibrary{trees}

\theoremstyle{plain}
\newtheorem{theor10}{Theorem}
\newenvironment{theor1}
  {\pushQED{\qed}\begin{theor10}}
  {\popQED\end{theor10}}
\newtheorem{lem0}{Lemma}[section]
\newenvironment{lem}
  {\pushQED{\qed}\begin{lem0}}
  {\popQED\end{lem0}}
\theoremstyle{plain}
\newtheorem{theor0}[lem0]{Theorem}
\newenvironment{theor}
  {\pushQED{\qed}\begin{theor0}}
  {\popQED\end{theor0}}
\newtheorem{prop0}[lem0]{Proposition}
\newenvironment{prop}
  {\pushQED{\qed}\begin{prop0}}
  {\popQED\end{prop0}}
\newtheorem{cor0}[lem0]{Corollary}
\newenvironment{cor}
  {\pushQED{\qed}\begin{cor0}}
  {\popQED\end{cor0}}

\numberwithin{equation}{section}

\newcommand{\e}{\varepsilon}
\newcommand{\R}{\mathbb R}
\newcommand{\Z}{\mathbb Z}

\newcommand{\C}{\mathbb C}
\newcommand{\D}{\mathbb D}
\newcommand{\T}{\mathbb T}
\newcommand{\Pc}{\mathcal P}
\newcommand{\Dc}{\mathcal D}
\newcommand{\Hc}{\mathcal H}
\newcommand{\Id}{\operatorname{Id}}
\newcommand{\LB}{{\operatorname{LB}}}
\newcommand{\LLB}{{\operatorname{LLB}}}
\newcommand{\loc}{{\operatorname{loc}}}
\newcommand{\Ld}{\operatorname{L}}
\newcommand{\step}[1]{\noindent \textit{Step} #1.}

\newcommand{\pv}{\operatorname{p.}\!\operatorname{v.}}

\usepackage[colorlinks,linkcolor=black,citecolor=black,urlcolor=black]{hyperref}

\title[Lenard--Balescu correction to mean-field theory]{Lenard--Balescu correction to mean-field theory}

\author[M. Duerinckx]{Mitia Duerinckx}
\address[Mitia Duerinckx]{Laboratoire de Mathématiques d'Orsay, UMR 8628, Université Paris-Sud, F-91405 Orsay, France \& Universit\'e Libre de Bruxelles, Département de Mathématique, Brussels, Belgium}
\email{mitia.duerinckx@u-psud.fr}
\author[L. Saint-Raymond]{Laure Saint-Raymond}
\address[Laure Saint-Raymond]{École Normale Supérieure de Lyon, UMR 5669, Unité de Mathématiques Pures et Appliquées, F-69342 Lyon, France}
\email{laure.saint-raymond@ens-lyon.fr}

\begin{document}
\selectlanguage{english}

\maketitle

\begin{abstract}
In the mean-field regime, the evolution of a gas of $N$ interacting particles is governed in first approximation by a Vlasov type equation with a self-induced force field.
This equation is conservative and describes return to equilibrium only in the very weak sense of Landau damping.
However, the first correction to this approximation is given by the Lenard--Balescu operator, which dissipates entropy on the very long timescale $O(N)$.
In this paper, we show how one can derive rigorously this correction on intermediate timescales (of order $O(N^r)$ for $r <1$), close to equilibrium.
\end{abstract}


\section{Introduction}
\subsection{General overview}\label{sec:overview}
Consider the dynamics of a system of $N$ classical particles in the torus $\T^d$ as given by Newton's equations of motion,
\begin{equation}\label{eq:part-dyn}
\frac{d}{dt}x_{j,N}=v_{j,N},\qquad\frac{d}{dt}v_{j,N}=-\frac1N\sum_{1\le l\le N\atop l\ne j}\nabla V(x_{j,N}-x_{l,N}),\qquad1\le j\le N,
\end{equation}
where $\{(x_{j,N},v_{j,N})\}_{j=1}^N$ denotes the set of positions and velocities of the particles in the phase space $\D:=\T^d\times\R^d$ and where $V:\T^d\to\R$ is a long-range interaction potential.
In the mean-field scaling that is considered here, the strength of each binary interaction is~$O(\frac1N)$, so that the total force exerted on a particle by all the others is of order~$O(1)$.
In terms of a probability density $F_N$ on the $N$-particle phase space $\D^N=(\T^d\times\R^{d})^N$, these equations are equivalent to the following Liouville equation,
\begin{align}\label{eq:Liouville}
\partial_tF_N+\sum_{j=1}^Nv_j\cdot\nabla_{x_j}F_N=\frac1N\sum_{1\le j\ne l\le N}\nabla V(x_j-x_l)\cdot\nabla_{v_j}F_N,
\end{align}
where particles are assumed to be exchangeable, hence $F_N$ is symmetric in its $N$ variables $z_j:=(x_j,v_j)\in\D$, $1\le j\le N$.
In the large-$N$ limit, rather than describing the whole set of individual particle trajectories in this $N$-body problem, one looks for an ``averaged'' description of the system: one may typically focus on the evolution of one ``typical'' particle in the system, as described by the first marginal of $F_N$,
\[F_{N}^{1;t}(z):=\int_{\D^{N-1}}F_N^{t}(z,z_2\ldots,z_N)\,dz_{2}\ldots dz_N.\]
Neglecting $2$-particle correlations (the so-called Boltzmann's chaos assumption) formally leads to the following mean-field approximation: $F_{N}^1$ is expected to remain close to the solution $H$ of the Vlasov equation,
\begin{align}\label{eq:Vlasov}
\partial_tH+v\cdot\nabla_xH=(\nabla V\ast\mu_H)\cdot\nabla_vH,\qquad\mu_H(x):=\int_{\R^d}H (x,v)\,dv.
\end{align}
We refer to~\cite{Golse-rev} for a review of rigorous results on this well-travelled topic.
Next, we may investigate the next-order correction to this mean-field approximation.
This is particularly relevant in the case of spatially homogeneous systems: if the mean-field density is spatially homogeneous, $H=H(v)$, it remains constant over time in view of~\eqref{eq:Vlasov}, hence the mean-field regime is trivial and the next-order correction becomes the relevant leading order. In that setting, we naturally focus on the velocity density of the typical particle,
\[f_N^{1;t}(v):=\int_{\T^d}F_N^{1;t}(x,v)\,dx.\]
In agreement with Bogolyubov's theory of non-equilibrium statistical mechanics~\cite{Bogolyubov-46},
while the mean-field approximation was obtained by neglecting particle correlations, the next-order correction precisely amounts to taking into account those $2$-particle correlations (or ``collisions''), only neglecting higher-order correlations.
Correlations should lead to irreversible effects and actually create a non-Markovian correction of order $O(\frac1N)$. This contribution is expected to become $O(1)$ only on the relevant long timescale $t\sim N$, and memory effects to vanish on this long timescale, hence giving rise to a simple collisional operator.
More precisely, the time-rescaled velocity density $f_{N}^{1;Nt}$ is predicted to remain close to the solution~$f$ of the following so-called Lenard--Balescu kinetic equation,
\begin{align}\label{eq:LB}
\partial_tf=\LB(f):=\nabla\cdot\int_{\R^d}B(v,v-v_*;\nabla f)\,\big(f_*\nabla f-f\nabla_*f_*\big)\,dv_*,
\end{align}
with the notation $f=f(v)$, $f_*=f(v_*)$, $\nabla=\nabla_v$, and $\nabla_*=\nabla_{v_*}$,
where the collision kernel is given by
\begin{align}\label{eq:LB-kernel}
B(v,v-v_*;\nabla f):=\sum_{k\in2\pi\Z^d}(k\otimes k)\,\pi\widehat V(k)^2\tfrac{\delta(k\cdot(v-v_*))}{|\e(k,k\cdot v;\nabla f)|^2},
\end{align}
and the dispersion function is
\begin{align}\label{eq:LB-disp}
\e(k,k\cdot v;\nabla f):=1+\widehat V(k)\int_{\R^d}\tfrac{k\cdot\nabla f(v_*)}{k\cdot(v-v_*)-i0}\,dv_*,
\end{align}
where henceforth the notation $\pm i0$ means that we take the limit $\alpha\downarrow0$ of the same quantity with $\pm i\alpha$.
Note that this  function $\e$ appears naturally when computing the resolvent of the Vlasov operator linearized around the equilibrium $f= f(v)$. In particular, the equilibrium is linearly stable as long as $\e$ does not vanish, which is referred to as the Penrose criterion~\cite{Penrose}; this is in particular the case for Maxwellian equilibria.
In dimension $d=1$, the Lenard-Balescu operator vanishes identically: the system then undergoes a kinetic blocking and is expected to evolve only on the longer timescale $t\sim N^2$ under the effect of $3$-particle correlations, cf.~\cite{FBOC-19,FCP-20}.

\medskip
The Lenard--Balescu equation (\ref{eq:LB}) was formally derived in the early 60s independently by Guernsey~\cite{Guernsey-60,Guernsey-62}, Lenard~\cite{Lenard-60}, and Balescu~\cite{Balescu-60,BT-61}. The former two derivations rely on the earlier work by Bogolyubov~\cite{Bogolyubov-46} and proceed by truncating the so-called BBGKY hierarchy of equations for correlation functions and computing the Markovian limit. This derivation, as presented e.g.~in~\cite[Appendix~A]{Nicholson}, is our main inspiration in the sequel. Balescu's derivation builds instead on Prigogine's theory of irreversibility~\cite{Prigogine-61,Balescu-63} by means of a diagrammatic approach.
Another derivation was later proposed by Klimontovich~\cite{Klimontovich-67} (see also~\cite[Section~51]{LL-vol10}) based on fluctuations of the empirical measure associated with the particle dynamics~\eqref{eq:part-dyn}.

\medskip
At a formal level, the Lenard--Balescu equation~\eqref{eq:LB} preserves mass, momentum, and kinetic energy, and satisfies an $H$-theorem,
\[\partial_t\int_{\R^d} f\log f=-\iint_{\R^d\times\R^d}\big((\nabla-\nabla_*)\sqrt{ff_*}\big)\cdot B(v,v-v_*;\nabla f)\big((\nabla-\nabla_*)\sqrt{ff_*}\big)~\le~0,\]
hence it should describe relaxation to Maxwellian equilibrium on the long timescale $t\gg N$.
For a thorough discussion of the relevance of this equation in plasma physics, we refer e.g.~to~\cite[Section~11.11]{Krall-Trivelpiece}, \cite[Section~47]{LL-vol10}, \cite[Chapter~5]{Nicholson}, or \cite[Part~1]{Decoster}.
A key feature of this equation is to take into account collective screening effects in form of the nonlocal nonlinearity of the kernel~\eqref{eq:LB-kernel} via the dispersion function $\e$. As such, it can be viewed as a corrected version of the celebrated (phenomenological) Landau equation, which indeed amounts to neglecting collective effects, that is, setting $\e\equiv1$.

\medskip
From the mathematical viewpoint, this dynamical screening in form of the full nonlinearity of the kernel~\eqref{eq:LB-kernel} makes the study of the Lenard--Balescu equation~\eqref{eq:LB} reputedly difficult.
Even local well-posedness remains an open problem, while the only rigorous results concern the relaxation of the linearized evolution at the Maxwellian equilibrium~\cite{Merchant-Liboff-73,Strain-07}.

\medskip
In this context, any rigorous derivation from particle dynamics~\eqref{eq:part-dyn} has remained elusive, apart from some partial attempts in~\cite{Lancellotti-10,VW-18} (see also~\cite{BPS-13,Winter-19} for corresponding attempts at the derivation of the Landau equation in the weak-coupling regime).
On top of well-posedness issues, the difficulty is mainly twofold: 
\begin{enumerate}[$\bullet$]
\item proving a priori estimates on correlation functions that are uniform up to the long timescale $t\sim N$. A major problem here is that regularity is not uniformly propagated by the Liouville equation (\ref{eq:Liouville}).
\item rigorously establishing the relaxation of the equation for the $2$-particle correlation function in order to get a closed kinetic equation for the density $f$. The problem is that  this relaxation holds only in a weak sense  corresponding to linear Landau damping for two typical particles, and therefore it is not clear that the error term (which converges only weakly to 0) would not contribute to the limiting dynamics.
\end{enumerate}

\medskip
In the present work,
we proceed to a fully rigorous analysis under the following two key simplifications:
\begin{enumerate}[(1)]
\item We focus on a linearized regime and consider two typical such settings:
\begin{itemize}
\item the evolution of a tagged particle in a thermal equilibrium background;
\item the evolution of tiny fluctuations around thermal equilibrium (more precisely, fluctuations of order $O(\frac1N)$, thus much smaller than typical fluctuations $O(\frac1{\sqrt N})$).
\end{itemize}
We mainly focus on the former setting, which is more relevant from a physical point of view, while the latter can be treated similarly and is briefly addressed in Appendix~\ref{app:lin-setting}.
\item We do not reach the relevant timescale $t\sim N$, but rather obtain the (suitably linearized) Lenard--Balescu evolution as a $O(\frac1{N^{1-r}})$ correction on the intermediate timescale $t\sim N^r$ for some small enough exponent $r>0$.
\end{enumerate}
The first simplification of course alleviates difficulties related to the nonlinear structure, including the well-posedness issue for~\eqref{eq:LB} and its rigorous derivation as a long-time limit. In addition, this linearization allows us to rigorously establish as in~\cite{BGSR-16} a weak version of chaos and neglect many-particle correlations uniformly in time, cf.~Section~\ref{sec:cum-estim}.
The second simplification is needed to avoid possibly delicate resonance questions related to singularities of long-time propagators (filamentation), as pointed out in Section~\ref{sec:echo}.
Note that the first simplification is relaxed in the companion article~\cite{D-19}, where a new approach to correlation estimates is developed away from equilibrium, but it requires to restrict to even shorter timescales $t\ll\log N$.

\subsection{Main result}
We focus on a linearized setting and consider the motion of a tagged particle in a bath of $N$ classical particles in the torus~$\T^d$ at thermal equilibrium (total number of particles is now $N+1$). More precisely, we denote by $z_{0,N}:=(x_{0,N},v_{0,N})\in\D$ the position and velocity of the tagged particle, and by $\{z_{j,N}:=(x_{j,N},v_{j,N})\}_{j=1}^N$ the positions and velocities of the $N$ background particles, the motion of which is governed by Newton's equations of motion,
\begin{equation}\label{eq:Newton-re}
\frac{d}{dt}x_{j,N}=v_{j,N},\qquad \frac{d}{dt}v_{j,N}=-\frac1{N}\sum_{0\le l\le N\atop l\ne j}\nabla V(x_{j,N}-x_{l,N}),\qquad0\le j\le N,
\end{equation}
and we assume that the background particles are initially at thermal equilibrium, that is, the initial positions and velocities  $\{z_{j,N}^\circ:=(x_{j,N}^\circ,v_{j,N}^\circ)\}_{j=0}^N$ are distributed according to the following probability density on the $(N+1)$-particle phase space $\D^{N+1}=(\T^d\times\R^d)^{N+1}$,
\begin{equation}\label{eq:tagged-initial}
F_N^\circ(z_0,\ldots,z_N)\,:=\,f^\circ(v_0)\,M_{N,\beta}(z_1,\ldots,z_N),
\end{equation}
where $f^\circ:\R^d\to\R^+$ denotes the initial velocity distribution of the tagged particle (assumed to be spatially homogeneous), and where the equilibrium background density $M_{N,\beta}$ is given by the Gibbs measure
\begin{equation}\label{eq:def-MNbeta}
M_{N,\beta}(z_1,\ldots,z_N)\,:=\,Z_{N,\beta}^{-1}\,e^{-\frac\beta2\sum_{j=1}^N|v_j|^2-\frac\beta{2N}\sum_{1\le j\ne l\le N}V(x_j-x_l)},
\end{equation}
with normalization factor $Z_{N,\beta}$ and with fixed inverse temperature $\beta\in(0,\infty)$.
While the mean-field force exerted by the equilibrium background
vanishes for symmetry reasons, the tagged particle is expected to display a Brownian trajectory in velocity space on the timescale $t\sim N$ and to progressively acquire the Maxwellian velocity distribution as the background itself,
\[M_\beta(v)\,:=\,(\tfrac\beta{2\pi})^\frac d2e^{-\frac\beta2|v|^2}.\]
More precisely, as first predicted in~\cite{TGM-64} (see also~\cite{Piasecki-81,PS-87}), its time-rescaled velocity distribution $f_N^{1;Nt}$ is expected to remain close to the solution $f$ of the following Fokker-Planck equation, which is obtained by replacing the time-dependent distribution~$f_*$ in the Lenard--Balescu equation~\eqref{eq:LB} by the velocity distribution $M_\beta$ of the equilibrium background,
\begin{equation}\label{eq:FP}
\partial_tf=\nabla\cdot \big(D_\beta(\nabla f+\beta vf)\big),
\end{equation}
in terms of the diffusion tensor
\[D_\beta(v):=\sum_{k\in2\pi\Z^d\setminus\{0\}}(k\otimes k)\,\tfrac{\pi\widehat V(k)^2}{|k||\e_\beta(k,k\cdot v)|^2} (\tfrac{\beta}{2\pi})^\frac 12e^{-\frac\beta{2}(v\cdot\frac k{|k|})^2},\]
with
\begin{align*}
\e_\beta(k,k\cdot v):=1-\widehat V(k)\int_{\R^d}\tfrac{\beta k\cdot v_*}{k\cdot(v-v_*)-i0}\,M_\beta(v_*)\,dv_*.
\end{align*}
Equation~\eqref{eq:FP} satisfies an $H$-theorem,
\[\partial_t\int_{\R^d}M_\beta\,\big|\tfrac f{M_\beta}\big|^2\,=\,-2\int_{\R^d}M_\beta\,\nabla(\tfrac f{M_\beta})\cdot D_\beta\nabla(\tfrac f{M_\beta})\,\le\,0,\]
hence it describes relaxation of the tagged particle velocity distribution to the Maxwellian $M_\beta$.
The following main result provides a fully rigorous derivation of~\eqref{eq:FP} starting from particle dynamics, although only on an intermediate timescale $t\sim N^r$ with $r<1$ small enough. The reason for this limitation is that we do not manage to rule out possible resonant effects as described in Section~\ref{sec:echo}.
As the velocity distribution does not evolve on this intermediate timescale, we simply obtain the Fokker-Planck operator applied to the initial data, instead of a genuine evolution equation.
The case of tiny fluctuations around thermal equilibrium is similarly addressed in Appendix~\ref{app:lin-setting},
while the nonlinear setting of Section~\ref{sec:overview} is treated in the companion article~\cite{D-19}. A similar consistency result for the Landau equation in the weak-coupling regime was obtained in~\cite{BPS-13}.

\begin{theor1}[Tagged particle in equilibrium bath]\label{th:main}
Let $d\ge1$ and assume that
\begin{itemize}
\item the interaction potential $V:\T^d\to\R$ is even, positive definite (that is, $\widehat V\ge0$), and smooth enough (that is, $V\in W^{4,\infty}(\T^d)$);
\item the temperature of the initial background~\eqref{eq:def-MNbeta} is large enough (that is, $\beta\|V\|_{\Ld^\infty}\le \frac1{C_0}$ for some large enough universal constant $C_0$);
\item the initial velocity distribution has Maxwellian decay (that is, $\frac{f^\circ}{M_\beta}\in\Ld^\infty(\R^d)$).
\end{itemize}
Then, for $0<r<\frac1{18}$, the velocity distribution $f_N$ of the tagged particle satisfies on the timescale $t\sim N^r$,
\[\lim_{N\uparrow\infty}N(\partial_tf_N^{1;t})|_{t=N^r\tau}\,=\,\nabla\cdot \big(D_\beta(\nabla f^\circ+\beta vf^\circ)\big),\]
as a function of $(\tau,v)$ in the weak sense of $\Dc'(\R^+\times\R^d)$.
\end{theor1}

\noindent
{\bf Outline of the paper.}
The key tool that we use in this paper is the set of cumulants of the probability density $F_N$, which encodes the correlations within clusters of particles of different sizes. In the mean-field regime, these correlations are expected to be small, all the more that the size of the cluster is large. Section~\ref{sec:cum} provides global-in-time a priori estimates on these cumulants, as well as their evolution equations. Then Section~\ref{sec:echo} explains the strategy to get asymptotically a closed equation on the first cumulant, which takes form of the Fokker-Planck equation~\eqref{eq:FP}. The main difficulty at this stage is to establish the stability of this equation under the uniform cumulant estimates, which have very low regularity and do not prevent oscillations or resonances. We therefore separate the proof of the convergence into two steps. The first one, in Section~\ref{sec:intermediate}, shows how one can truncate the cumulant hierarchy at order~$3$, controlling higher-order contributions by means of the uniform cumulant estimates, combined with resolvent estimates for the linearized Vlasov equation. Next, in Section~\ref{sec:eqns-lim}, we identify the Lenard--Balescu operator, combining the first two equations of the cumulant hierarchy and computing resolvents explicitly by complex deformations.

\bigskip
\noindent
{\bf Notation.}
\begin{enumerate}[$\bullet$]
\item We denote by $C\ge1$ any constant that only depends on $d$, on $\|V\|_{W^{4,\infty}(\T^d)}$, and on $\|{f^\circ}/{M_\beta}\|_{\Ld^\infty(\R^d)}$. We use the notation $\lesssim$ (resp.~$\gtrsim$) for $\le C\times$ (resp.\@ $\ge\frac1C\times$) up to such a multiplicative constant $C$. We write $\simeq$ when both $\lesssim$ and $\gtrsim$ hold. We add subscripts to $C,\lesssim,\gtrsim,\simeq$ in order to indicate dependence on other parameters. We denote by $O(K)$ a quantity that is bounded by $CK$.
\item For $0\le m\le n$, we set $[m,n]:=\{m,\ldots,n\}$ and we use the abbreviation $[m]:=[1,m]=\{1,\ldots, m\}$. For an index set $\sigma=\{i_1,\ldots,i_k\}$ we write $z_{\sigma}:=(z_{i_1},\ldots,z_{i_k})$.
\item For $a,b\in\R$ we write $a\wedge b:=\min\{a,b\}$ and $\langle a\rangle:=(1+a^2)^{1/2}$.
\item We use the following shorthand notation for averaging with respect to the Maxwellian distribution,
\begin{equation*}
\langle G(v)\rangle_v:=\int_{\R^d}G(v)\,M_\beta(v)\,dv.
\end{equation*}
\end{enumerate}

\section{Cumulant formalism and BBGKY hierarchy}\label{sec:cum}
This section is devoted to notation and basic results for the analysis of many-particle correlations in the BBGKY formalism for a tagged particle in an equilibrium bath~\eqref{eq:Newton-re}--\eqref{eq:def-MNbeta}.
The label $0$ is reserved for the tagged particle, while the labels $1,\dots,N$ are used for the (exchangeable) background particles.
We denote here by $F_N$ the probability density  on the $(N+1)$-particle phase space $\D^{N+1}$, which satisfies the following Liouville equation,
\begin{align}\label{eq:Liouville2}
\partial_tF_N+\sum_{j=0}^Nv_j\cdot\nabla_{x_j}F_N=\frac1N\sum_{0\le j\ne l\le N}\nabla V(x_j-x_l)\cdot\nabla_{v_j}F_N.
\end{align}
Note that $F_N$ is symmetric in its $N$ last entries as it is initally, cf.~\eqref{eq:tagged-initial}, embodying the exchangeability of the $N$ background particles. 
In order to prove Theorem~\ref{th:main}, we proceed to a linear cumulant expansion of $F_N$, which nicely splits correlations between subsets of particles. Cumulants satisfy a system of coupled equations, which is a variant of the celebrated BBGKY hierarchy.
As in~\cite{BGSR-16}, in the present linear setting, suboptimal a priori estimates on cumulants can be deduced by symmetry and yield a weak version of chaos, which allows a rigorous truncation of the hierarchy, cf.~Section~\ref{sec:intermediate}.

\subsection{Cumulant expansion and estimates}\label{sec:cum-estim}
Cumulants of a probability density are polynomial combinations of marginals that encode many-particle correlations and allow to recover the original distribution in form of a cluster expansion.
In the present setting, as the distribution of the $N$ background particles is close to equilibrium $M_{N,\beta}$, we naturally linearize the definition of cumulants. While the equilibrium contains particle correlations that are difficult to work with in practice, we rather choose to linearize the definition of cumulants around the simpler Maxwellian proxy $M_\beta^{\otimes N}$.
Up to errors due to this simplification, the linear cumulant $G_N^{m+1}$ below describes the correlation of the tagged particle with $m$ background particles.
Since $F_N$ is symmetric in its $N$ last entries, we note that the marginal $F_N^{m+1}$ and cumulant $G_N^{m+1}$ are similarly symmetric in their $m$ last entries.
The proof of the following key result is straightforward and can be found in~\cite[Proposition~4.2]{BGSR-16}.

\begin{lem}[\cite{BGSR-16}]\label{lem:cumulants}
For $0\le m\le N$, let the $(m+1)$th-order marginal $F_N^{m+1}:\D^{m+1}\to\R^+$ of $F_N$ be defined as
\[F_{N}^{m+1}(z_{[0,m]})\,:=\,\int_{\D^{N-m}}F_N(z_{[0,N]})\,dz_{[m+1,N]},\]
and define the corresponding cumulant $G_N^{m+1}:\D^{m+1}\to\R$ as
\begin{equation}\label{eq:def-cum}
G_N^{m+1}(z_{[0,m]})\,:=\,\sum_{j=0}^m(-1)^{m-j}\sum_{\sigma\in\mathfrak P_{j}^m}\tfrac{F_{N}^{j+1}}{M_\beta^{\otimes(j+1)}}(z_0,z_\sigma),
\end{equation}
where $\mathfrak P_{j}^m$ denotes the set of all subsets of $[m]$ with $j$ elements.
Then, the following two properties hold.
\begin{enumerate}[(i)]
\item \emph{Cluster expansion:} for all $0\le m\le N$,
\[F_{N}^{m+1}(z_{[0,m]})\,=\,M_\beta^{\otimes (m+1)}(v_{[0,m]})\sum_{j=0}^m\sum_{\sigma\in\mathfrak P_j^m}G_N^{j+1}(z_0,z_\sigma).\]
\item \emph{Orthogonality:} for all $1\le m\le N$ there holds $\int_\D G_N^{m+1}(z_{[0,m]})\,M_\beta(v_l)\,dz_l=0$ for each $l\in[m]$, hence
\[\int_{\D^N}\tfrac{|F_N|^2}{M_\beta^{\otimes (N+1)}}\,=\,\sum_{m=0}^N\binom Nm\int_{\D^{m+1}}|G_N^{m+1}|^2M_\beta^{\otimes(m+1)}.\qedhere\]
\end{enumerate}
\end{lem}

As the above shows, for $0\le m\le N$, the so-defined cumulant $G_N^{m+1}$ is an element of the following Hilbert space,
\begin{multline*}
\Ld^2_\beta(\D^{m+1})\,:=\,\bigg\{G\in\Ld^2_\loc(\D^{m+1}):\int_{\D^{m+1}}|G|^2M_\beta^{\otimes(m+1)}<\infty,\\
\text{ $G$ symmetric in the last $m$ entries}
\bigg\},
\end{multline*}
endowed with the Hilbert norm
\[\|G\|_{\Ld^2_\beta(\D^{m+1})}^2:=\langle G,G\rangle_{\Ld^2_\beta(\D^{m+1})},\qquad\langle H,G\rangle_{\Ld^2_\beta(\D^{m+1})}:=\int_{\D^{m+1}}\overline H\,G\,M^{\otimes (m+1)}_\beta,\]
and for $1\le p\le\infty$ we similarly define $\Ld^p_\beta(\D^{m+1})$ with
\[\|G\|_{\Ld^p_\beta(\D^{m+1})}:=\Big(\int_{\D^{m+1}}|G|^pM_\beta^{\otimes(m+1)}\Big)^\frac1p.\]
As in~\cite{BGSR-16}, the use of symmetry in Lemma~\ref{lem:cumulants}(ii) in form of combinatorial factors leads to a priori estimates on cumulants, which describe some decorrelation between particles and can thus be seen as a weak version of chaos.
An important feature is that these estimates further hold uniformly in time, therefore playing a key role for rigorous long-time analysis.
Note however that they are only suboptimal: $G_N^2$ is for instance expected to be of order $O(\frac1N)$ instead of $O(\frac1{N^{1/2}})$,
but these estimates serve as a starting point to be improved a posteriori.

\begin{lem}[Time-uniform a priori estimates on cumulants]\label{lem:cum-ap}
If for some $q>1$ there hold
\[\tfrac{f^\circ}{M_\beta}\in\Ld^{2q}_\beta(\D)\qquad\text{and}\qquad\beta\|V\|_{\Ld^\infty(\T^d)}<\tfrac{4(q-1)}{2q-1},\]
then for all $0\le m\le N$ we have
\[\sup_{t\ge0}\|G_N^{m+1;t}\|_{\Ld^2_\beta(\D^{m+1})}\,\lesssim_{m,q,\beta}\, N^{-\frac{m}2}\big\|\tfrac{f^\circ}{M_\beta}\big\|_{\Ld^{2q}_\beta(\D)}.\qedhere\]
\end{lem}

\begin{proof}
Next to~\eqref{eq:def-MNbeta}, define the full $(N+1)$-particle Gibbs measure
\begin{equation*}
\widetilde M_{N,\beta}(z_0,z_{[N]})\,:=\,\widetilde Z_{N,\beta}^{-1}\,e^{-\frac\beta2\sum_{j=0}^N|v_j|^2-\frac\beta{2N}\sum_{0\le j\ne l\le N}V(x_j-x_l)},
\end{equation*}
which is a global equilibrium of the Liouville equation~\eqref{eq:Liouville2}. Hence, we find for $1\le q<\infty$,
\[\partial_t\int_{\D^N}\tfrac{|F_N^t|^{2q}}{|\widetilde M_{N,\beta}|^{2q-1}}=0.\]
For $1<q<\infty$, setting $q':=\frac{q}{q-1}$, Lemma~\ref{lem:cumulants}(ii), Hölder's inequality, and this equilibrium property lead to
\begin{eqnarray}\label{eq:bound-GNk-pr}
\|G_N^{m+1;t}\|_{\Ld^2_\beta(\D^{m+1})}^2&\le&\binom Nm^{-1}\int_{\D^N}\tfrac{|F_N^t|^2}{M_\beta^{\otimes (N+1)}}\nonumber\\
&\le&\binom Nm^{-1}\Big(\int_{\D^{N+1}}\tfrac{|F_N^\circ|^{2q}}{|\widetilde M_{N,\beta}|^{2q-1}}\Big)^\frac1q\Big(\int_{\D^{N+1}}\tfrac{|\widetilde M_{N,\beta}|^{q'+1}}{|M_\beta^{\otimes (N+1)}|^{q'}}\Big)^\frac{1}{q'},
\end{eqnarray}
and it remains to estimate the two right-hand side factors.
Factoring out Maxwellians, we split
\begin{equation}\label{eq:split-MN}
M_{N,\beta}=M_\beta^{\otimes N}M_{N,\beta}',\qquad\widetilde M_{N,\beta}=M_\beta^{\otimes (N+1)}\widetilde M_{N,\beta}',
\end{equation}
in terms of
\begin{eqnarray*}
M_{N,\beta}'(x_{[N]})&:=&(Z_{N,\beta}')^{-1}e^{-\frac\beta{2N}\sum_{1\le j\ne l\le N}V(x_j-x_l)},\\
\widetilde M_{N,\beta}'(x_0,x_{[N]})&:=&(\widetilde Z_{N,\beta}')^{-1}e^{-\frac\beta{2N}\sum_{0\le j\ne l\le N}V(x_j-x_l)},
\end{eqnarray*}
where $Z_{N,\beta}'$ and $\widetilde Z_{N,\beta}'$ are the corresponding normalization factors on $(\T^d)^N$ and $(\T^d)^{N+1}$, respectively.
The definition~\eqref{eq:tagged-initial} of $F_N^\circ$ yields
\begin{eqnarray*}
\Big(\int_{\D^{N+1}}\tfrac{|F_N^\circ|^{2q}}{|\widetilde M_{N,\beta}|^{2q-1}}\Big)^\frac1q&=&\big\|\tfrac{f^\circ}{M_\beta}\big\|_{\Ld^{2q}_\beta(\D)}^2\Big(\int_{(\T^d)^{N+1}}\tfrac{|M_{N,\beta}'|^{2q}}{|\widetilde M_{N,\beta}'|^{2q-1}}\Big)^\frac1q\\
&\le&e^{2\beta\|V\|_{\Ld^\infty(\T^d)}}\big\|\tfrac{f^\circ}{M_\beta}\big\|_{\Ld^{2q}_\beta(\D)}^2\Big(\tfrac{\widetilde Z_{N,\beta}'}{Z_{N,\beta}'}\Big)^{1+\frac1{q'}},
\end{eqnarray*}
and similarly,
\begin{align*}
\int_{\D^{N+1}}\tfrac{|\widetilde M_{N,\beta}|^{q'+1}}{|M_\beta^{\otimes (N+1)}|^{q'}}\,=\,\tfrac{Z'_{N,(q'+1)\beta}}{(Z'_{N,\beta})^{q'+1}}.
\end{align*}
Inserting these estimates into~\eqref{eq:bound-GNk-pr} leads to
\begin{equation}\label{eq:bound-GNm-pre}
\|G_N^{m+1;t}\|_{\Ld^2_\beta(\D^{k+1})}^2\le\binom Nm^{-1}e^{2\beta\|V\|_{\Ld^\infty(\T^d)}}\big\|\tfrac{f^\circ}{M_\beta}\big\|_{\Ld^{2q}_\beta(\D)}^2\tfrac{(\widetilde Z_{N,\beta}')^{1+\frac1{q'}}(Z'_{N,(q'+1)\beta})^{\frac1{q'}}}{(Z'_{N,\beta})^{2(1+\frac1{q'})}}.
\end{equation}
The problem is thus reduced to checking that the last factor is bounded uniformly in~$N$ for~$\beta$ small enough.
Rather than going through a direct tedious computation of the partition functions, we appeal to standard large deviation theory for particle systems: it follows from
e.g.~\cite[Theorem~B(ii)]{Benarous-Brunaud-90} that $e^{Nc_{\beta}}Z'_{N,\beta}\to Z'_{\beta}\in(0,\infty)$ and $e^{Nc_{\beta}}\widetilde Z'_{N,\beta}\to\widetilde Z'_{\beta}\in(0,\infty)$ as $N\uparrow\infty$, where the constant $c_\beta$ is characterized by
\[c_\beta:=\inf\bigg\{\int_{\T^d}\log\mu\,d\mu+\frac\beta2\iint_{\T^d\times\T^d} V(x-y)\,d\mu(x)d\mu(y)\,:\,\mu\in\Pc(\T^d)\bigg\},\]
whenever this infimum is reached at some $\mu_0$ and is non-degenerate in the sense that the operator $\Sigma$ on $\Ld^2(\T^d,\mu_0)$ defined by $\Sigma f (y) := 2 \int_{\T^d} f(x)  V(x-y)\,d\mu_0(x)$
does not have 1 as an eigenvalue.
As $\int_{\T^d}V=\widehat V(0)\ge0$, the Csisz\'ar-Kullback-Pinsker inequality yields
\begin{eqnarray*}
\lefteqn{\int_{\T^d}\log\mu\,d\mu+\frac\beta2\iint_{\T^d\times\T^d} V(x-y)\,d\mu(x)d\mu(y)}\\
&=&\int_{\T^d}\log\mu\,d\mu+\frac\beta2\iint_{\T^d\times\T^d} V(x-y)\,(d\mu(x)-dx)(d\mu(y)-dy)+ \frac\beta2 \widehat V(0) \\
&\ge &\int_{\T^d}\log\mu\,d\mu-\frac\beta2\|V\|_{\Ld^\infty(\T^d)}\|\mu-1\|_{\operatorname{TV}}^2+ \frac\beta2 \widehat V(0)\\
&\ge& \big(2-\tfrac\beta2\|V\|_{\Ld^\infty(\T^d)}\big)\|\mu-1\|_{\operatorname{TV}}^2+\frac\beta2 \widehat V(0).
\end{eqnarray*}
Hence, for $\beta\|V\|_{\Ld^\infty(\T^d)}<4$, there holds $c_\beta=\frac\beta2 \widehat V(0)$ and the infimum is indeed nondegenerate, which entails $e^{N\frac\beta2 \widehat V(0)}Z'_{N,\beta}\to Z'_{\beta}$ and $e^{N\frac\beta2 \widehat V(0)}\widetilde Z'_{N,\beta}\to \widetilde Z'_{\beta}$ and implies that the last factor in~\eqref{eq:bound-GNm-pre} is uniformly bounded as $N\uparrow\infty$.
\end{proof}

Next, we compute the first three initial cumulants, based on the special form~\eqref{eq:tagged-initial} of the initial data $F_N^\circ$.

\begin{lem}[Initial cumulants]\label{lem:initial}
The initial first two cumulants are
\[G_N^{1;\circ}(z)=\tfrac{f^\circ}{M_\beta}(v)=:g^\circ(v),\qquad G_N^{2;\circ}=0,\]
and the initial third cumulant can be written as $G_N^{3;\circ}(z_{[0,2]})=\tfrac1Ng^\circ(v_0)H_{N,\beta}^\circ(x_1-x_2)$ for some function $H_{N,\beta}^\circ$ with $\|H^\circ_{N,\beta}\|_{\Ld^2(\T^d)}\lesssim_\beta1$.
\end{lem}

\begin{proof}
The computation of the first marginal $F_N^{1;\circ}=f^\circ$ of $F_N^\circ$ is obvious from~\eqref{eq:tagged-initial}. We turn to the next marginals,
and we split $M_{N,\beta}$ as in~\eqref{eq:split-MN}.
By translation invariance of~$M'_{N,\beta}$ on the torus, we find
\[\tfrac{F_N^{2;\circ}}{M_\beta^{\otimes2}}(z_{[0,1]})\,=\,g^\circ(v_0)\int_{(\T^d)^{N-1}}M'_{N,\beta}(x_{[N]})\,dx_{[2,N]}=g^\circ(v_0),\]
hence by definition~\eqref{eq:def-cum},
\[G_{N}^{2;\circ}(z_{[0,1]})=\tfrac{F_N^{2;\circ}}{M_\beta^{\otimes2}}(z_0,z_1)-\tfrac{F_N^{1;\circ}}{M_\beta}(z_0)=0.\]
Next, the third marginal takes the form
\[\tfrac{F_N^{3;\circ}}{M_\beta^{\otimes3}}(z_{[0,2]})\,=\,g^\circ(v_0)\int_{(\T^d)^{N-2}}M'_{N,\beta}(x_{[N]})\,dx_{[3,N]},\]
hence by definition~\eqref{eq:def-cum},
\[G_N^{3;\circ}(z_{[0,2]})=g^\circ(v_0)\Big(\int_{(\T^d)^{N-2}}M'_{N,\beta}(x_{[N]})\,dx_{[3,N]}-1\Big).\]
By translation invariance, we can indeed write $G_N^{3;\circ}(z_{[0,2]})=\tfrac1Ng^\circ(v_0)H_{N,\beta}^\circ(x_1-x_2)$ for some function $H_{N,\beta}^\circ:\T^d\to\R$. Applying Lemma~\ref{lem:cum-ap} in the form
\[\|g^\circ\|_{\Ld^2_\beta(\R^d)}\|H^\circ_{N,\beta}\|_{\Ld^2(\T^d)}\,=\,\|NG_N^{3;\circ}\|_{\Ld^2_\beta(\D^3)}\,\lesssim_\beta\,\|g^\circ\|_{\Ld^\infty(\R^d)},\]
the conclusion follows.
\end{proof}

\subsection{BBGKY hierarchy}\label{sec:eqns-BBGKY}
As is classical, the Liouville equation~\eqref{eq:Liouville2} for $F_N$ is equivalent to the following BBGKY hierarchy of equations for the marginals, for $0\le m\le N$,
\begin{multline}\label{eq:BBGKY}
\partial_tF_{N}^{m+1}+\sum_{j=0}^mv_j\cdot\nabla_{x_j}F_{N}^{m+1}=\frac1N\sum_{0\le j\ne l\le m}\nabla V(x_j-x_l)\cdot\nabla_{v_j}F_{N}^{m+1}\\
+\tfrac{N-m}N\sum_{j=0}^m\int_{\D}\nabla V(x_j-x_{*})\cdot\nabla_{v_j}F_{N}^{m+2}(z_0,z_{[m]},z_*)\,dz_{*},
\end{multline}
with the convention $F_N^{N+2}:=0$.
Note that the first right-hand side term is of order $O(\frac{m^2}N)$ and is precisely the one that creates correlations between initially independent particles and deviates from the mean-field theory.
This hierarchy~\eqref{eq:BBGKY} can alternatively be written as a hierarchy of equations on the cumulants.

\begin{lem}[BBGKY hierarchy on cumulants]\label{lem:eqns}
For $0\le m\le N$, the cumulant $G_N^{m+1}$ satisfies
\[\partial_tG_N^{m+1}+iL_{m+1}G_N^{m+1}\,=\,M_{m+2}^{m+1}G_N^{m+2}+\tfrac1N\big(S_{m-1}^{m+1}G_N^{m-1}+S_m^{m+1}G_N^m+S_{m+1}^{m+1}G_N^{m+1}\big),\]
with the convention $G_N^{-1},G_N^0,G_N^{N+2}=0$, where we have set $L_{m+1}:=\sum_{j=0}^mL_{m+1}^{(j)}$ and
\begin{eqnarray*}
iL_{m+1}^{(j)}G_N^{m+1}&:=&v_j\cdot\nabla_{x_j}G_N^{m+1}+\mathds1_{j\ne0}\tfrac{N+1-m}N\beta v_j\cdot\int_{\D}\nabla V(x_j-x_{*})\,\\
&&\hspace{4cm}\times G_N^{m+1}(z_0,z_{[m]\setminus\{j\}},z_*)\,M_\beta(v_{*})\,dz_{*},\\
M_{m+2}^{m+1}G_N^{m+2}&:=&\tfrac{N-m}N\sum_{j=0}^m\int_{\D}\nabla V(x_j-x_{*})\cdot(\nabla_{v_j}-\beta v_j) G_N^{m+2}(z_0,z_{[m]},z_*)\,M_\beta(v_{*})\,dz_{*},\\
S_{m-1}^{m+1}G_N^{m-1}&:=&\sum_{1\le j\ne l\le m}\nabla V(x_j-x_l)\cdot(-\beta v_j)G_N^{m-1}(z_0,z_{[m]\setminus\{j,l\}}),\\
S_m^{m+1}G_N^{m}&:=&\sum_{j=0}^m\sum_{1\le l\le m\atop j\ne l}\nabla V(x_j-x_l)\cdot(\nabla_{v_j}-\beta v_j+\beta v_l)G_N^{m}(z_0,z_{[m]\setminus\{l\}})\\
&&-\sum_{1\le j\ne l\le m}\int_{\D}\nabla V(x_j-x_{*})\cdot(-\beta v_j)\,G_N^m(z_0,z_{[m]\setminus\{j,l\}},z_*)\,M_\beta(v_{*})\,dz_{*},\\
S_{m+1}^{m+1}G_N^{m+1}&:=&\sum_{0\le j\ne l\le m}\nabla V(x_j-x_l)\cdot(\nabla_{v_j}-\beta v_j)G_N^{m+1}\\
&&\hspace{-1.4cm}-\sum_{j=0}^m\sum_{1\le l\le m\atop j\ne l}\int_{\D}\nabla V(x_j-x_{*})\cdot(\nabla_{v_j}-\beta v_j) G_N^{m+1}(z_0,z_{[m]\setminus\{l\}},z_*)\,M_\beta(v_{*})\,dz_{*}.\qedhere
\end{eqnarray*}
\end{lem}

For later purposes, we give without proof the expression of the different above-defined operators in Fourier space: denoting by $k_j\in2\pi\Z^d$ the Fourier variable associated with $x_j\in\T^d$ and setting $\hat z_j:=(k_j,v_j)\in\widehat\D:=2\pi\Z^d\times\R^d$,
\begingroup\allowdisplaybreaks
\begin{eqnarray*}
i\widehat L_{m+1}^{(j)}\widehat G_N^{m+1}&:=&ik_j\cdot v_j\Big(\widehat G_N^{m+1}+\mathds1_{j\ne0}\tfrac{N+1-m}N\beta\widehat V(k_j)\langle \widehat G_N^{m+1}\rangle_{v_j}\Big),\\
\widehat M_{m+2}^{m+1}\widehat G_N^{m+2}&:=&\tfrac{N-m}N\sum_{j=0}^m\sum_{k_*\in2\pi\Z^d}ik_*\widehat V(k_*)\\
&&\hspace{1.4cm}\cdot(\nabla_{v_j}-\beta v_j) \big\langle\widehat G_N^{m+2}\big(\hat z_{[0,m]\setminus\{j\}},(k_j-k_*,v_j),(k_*,v_*)\big)\big\rangle_{v_*},\\
\widehat S_{m-1}^{m+1}\widehat G_N^{m-1}&:=&-\sum_{1\le j\ne l\le m}\delta(k_j+k_l)\,ik_j\cdot v_j\beta\widehat V(k_j)\widehat G_N^{m-1}(\hat z_0,\hat z_{[m]\setminus\{j,k\}}),\\
\widehat S_m^{m+1}\widehat G_N^{m}&:=&
-\sum_{j=0}^m\sum_{1\le l\le m\atop j\ne l}ik_l\widehat V(k_l)\cdot(\nabla_{v_j}-\beta v_j+\beta v_l)\,\widehat G_N^{m}\big(\hat z_{[0,m]\setminus\{j,l\}},(k_j+k_l,v_j)\big)\\
&&\hspace{1.4cm}+\sum_{1\le j\ne l\le m}ik_j\cdot v_j\beta\widehat V(k_j)\,\big\langle\widehat G_N^m\big(\hat z_0,\hat z_{[m]\setminus\{j,l\}},(k_j,v_*)\big)\big\rangle_{v_*},\\
\widehat S_{m+1}^{m+1}\widehat G_N^{m+1}&:=&\sum_{0\le j\ne l\le m}\sum_{k_*\in2\pi\Z^d}ik_*\widehat V(k_*)\\
&&\hspace{1.4cm}\cdot(\nabla_{v_j}-\beta v_j)\widehat G_N^{m+1}\big(\hat z_{[0,m]\setminus\{j,l\}},(k_j-k_*,v_j),(k_l+k_*,v_l)\big)\\
&&\hspace{-2.7cm}-\sum_{j=0}^m\sum_{1\le l\le m\atop j\ne l}\sum_{k_*\in2\pi\Z^d}ik_*\widehat V(k_*)\cdot(\nabla_{v_j}-\beta v_j) \big\langle\widehat G_N^{m+1}(\hat z_{[0,m]\setminus\{j,l\}},(k_j-k_*,v_j),(k_*,v_*)\big)\big\rangle_{v_*},
\end{eqnarray*}
\endgroup
where for shortness we abusively use
the implicit correct ordering of variables according to the labels in velocity.

\begin{proof}
[Proof of Lemma~\ref{lem:eqns}]
Combining the definition~\eqref{eq:def-cum} of cumulants in terms of marginals together with the equations of the BBGKY hierarchy~\eqref{eq:BBGKY} for marginals, we find
\[\partial_tG_N^{m+1}+\sum_{k=0}^mv_k\cdot\nabla_{x_k}G_N^{m+1}=A+B,\]
where we have set
\[A\,:=\,\sum_{k=0}^m(-1)^{m-k}\sum_{\sigma\in\mathfrak P_k^m}\tfrac1N\sum_{i,j\in\sigma\cup\{0\}\atop i\ne j}\nabla V(x_i-x_j)\cdot(\nabla_{v_i}-\beta v_i)\tfrac{F_{N}^{k+1}}{M_\beta^{\otimes (k+1)}}(z_0,z_\sigma),\]
and
\begin{multline*}
B:=\sum_{k=0}^m(-1)^{m-k}\big(\tfrac{N-k}N\big)\sum_{\sigma\in\mathfrak P_k^m}\sum_{i\in\sigma\cup\{0\}}\\
\times\int_{\D}\nabla V(x_i-x_{*})\cdot(\nabla_{v_i}-\beta v_i)\tfrac{F_{N}^{k+2}}{M_\beta^{\otimes(k+2)}}(z_0,z_\sigma,z_{*})\,M_\beta(v_{*})\,dz_*.
\end{multline*}
We consider the two right-hand side terms separately and we start with the first one.
Reorganizing the sums yields
\begin{align}\label{eq:rewr-A}
A\,=\,\tfrac1N\sum_{0\le i,j\le m\atop i\ne j}\nabla V(x_i-x_j)\cdot(\nabla_{v_i}-\beta v_i)\sum_{k=0}^m(-1)^{m-k}\sum_{\sigma\in\mathfrak P_k^m\atop i,j\in\sigma\cup\{0\}}\tfrac{F_N^{k+1}}{M_\beta^{\otimes (k+1)}}(z_0,z_\sigma).
\end{align}
Inserting the cluster expansion for marginals in terms of cumulants (cf.~Lemma~\ref{lem:cumulants}(i)), we get
\[T^m_{ij}:= \sum_{k=0}^m(-1)^{m-k}\sum_{\sigma\in\mathfrak P_k^m\atop i,j\in\sigma\cup\{0\}}\tfrac{F_N^{k+1}}{M_\beta^{\otimes(k+1)}}(z_0,z_\sigma)
= \sum_{k=0}^m(-1)^{m-k}\sum_{\sigma\in\mathfrak P_k^m\atop i,j\in\sigma\cup\{0\}} \sum_{l=0}^k\sum_{\tau\in\mathfrak P_l^\sigma}G_N^{l+1}(z_0,z_\tau).\]
First consider the case when $m\geq 2$ and $i,j \neq 0$. We may then decompose
\begin{eqnarray*}
T^m_{ij}
&=& \sum_{l=0}^m\sum_{\tau\in\mathfrak P_l^m \atop i,j \in \tau }G_N^{l+1}(z_0,z_\tau) \sum_{ k= l}^m (-1)^{m-k}\,\sharp\{ \sigma\in\mathfrak P_k^m:\tau \subset \sigma\}\\
 &&+ \sum_{l=0}^{m-1}\sum_{\tau\in\mathfrak P_l^m \atop i\in \tau, j\not\in \tau  }G_N^{l+1}(z_0,z_\tau) \sum_{ k= l}^m (-1)^{m-k}\,\sharp\{ \sigma\in\mathfrak P_k^m: \tau \cup \{j\} \subset \sigma\}\\
&&+\sum_{l=0}^{m-1}\sum_{\tau\in\mathfrak P_l^m \atop j \in \tau, i \not\in \tau  }G_N^{l+1}(z_0,z_\tau) \sum_{ k= l}^m (-1)^{m-k}\,\sharp\{ \sigma\in\mathfrak P_k^m: \tau  \cup \{i\} \subset \sigma\}\\
&&+ \sum_{l=0}^{m-2}\sum_{\tau\in\mathfrak P_l^m \atop i,j \not\in \tau }G_N^{l+1}(z_0,z_\tau) \sum_{ k= l}^m (-1)^{m-k}\,\sharp\{ \sigma\in\mathfrak P_k^m:\tau \cup \{i,j\}  \subset \sigma\}.
\end{eqnarray*}
Computing the cardinalities and using the identity $\sum_{ k= 0} ^ p (-1) ^{p-k} \binom{p}{k} = \delta_{p=0}$, we deduce 
\[T^m_{ij}
\,=\, G_N^{m+1}(z_0,z_{[m]}) + G_N^{m}(z_0,z_{[m]\setminus\{j \}})+  G_N^{m}(z_0,z_{[m]\setminus\{i \}}) + G_N^{m-1}(z_0,z_{[m]\setminus\{i ,j\}}).\]
When $i=0$ or $j=0$, two terms in the previous sum obviously disappear.
We conclude for all $0\le i,j\le m$ with $i\ne j$,
\begin{multline*}
T^m_{ij}\,=\,G_N^{m+1}(z_0,z_{[m]})
+\big(G_N^{m}(z_0,z_{[m]\setminus\{i\}})\,\mathds1_{i\ne0}+G_N^{m}(z_0,z_{[m]\setminus\{j\}})\,\mathds1_{j\ne0}\big)\mathds1_{m\ge1}\\
+G_N^{m-1}(z_0,z_{[m]\setminus\{i,j\}})\,\mathds1_{i,j\ne0}\,\mathds1_{m\ge2}.
\end{multline*}
Inserting this identity into~\eqref{eq:rewr-A}, since $\nabla V$ is odd, we find
\begin{multline*}
A\,=\,\tfrac1N\sum_{0\le i,j\le m\atop i\ne j}\nabla V(x_i-x_j)\cdot(\nabla_{v_i}-\beta v_i)G_N^{m+1}\\
+\mathds1_{m\ge1}\,\tfrac1N\sum_{0\le i\le m,1\le j\le m\atop i\ne j}\nabla V(x_i-x_j)\cdot(\nabla_{v_i}-\beta v_i+\beta v_j)G_N^{m}(z_0,z_{[m]\setminus\{j\}})\\
+\mathds1_{m\ge2}\,\tfrac1N\sum_{1\le i,j\le m\atop i\ne j}\nabla V(x_i-x_j)\cdot(-\beta v_i)G_N^{m-1}(z_0,z_{[m]\setminus\{i,j\}}).
\end{multline*}
We turn to the second term $B$.
Noting that $\int_\D\nabla V(x_i-x_{m+1})\,M_\beta(v_{m+1})\,dz_{m+1}=0$ and reorganizing the sums, we can rewrite
\begin{multline*}
B=\sum_{i=0}^m\int_{\D}\nabla V(x_i-x_{m+1})\cdot(\nabla_{v_i}-\beta v_i)\\
\times\sum_{k=0}^{m+1}(-1)^{m+1-k}\big(\tfrac{N+1-k}N\big)\sum_{\sigma\in\mathfrak P_{k}^{m+1}\atop i\in\sigma\cup\{0\}}\tfrac{F_{N}^{k+1}}{M_\beta^{\otimes(k+1)}}(z_0,z_\sigma)\,M_\beta(v_{m+1})\,dz_{m+1},
\end{multline*}
which we decompose as $B=B_0+B_1$ in terms of
\begin{eqnarray*}
B_0&:=&\sum_{i=0}^m\int_{\D}\nabla V(x_i-x_{m+1})\cdot(\nabla_{v_i}-\beta v_i)\\
&&\times\sum_{k=0}^{m+1}(-1)^{m+1-k}\big(\tfrac{N+1-k}N\big)\sum_{\sigma\in\mathfrak P_{k}^{m+1}}\tfrac{F_{N}^{k+1}}{M_\beta^{\otimes(k+1)}}(z_0,z_\sigma)\,M_\beta(v_{m+1})\,dz_{m+1},\\
B_1&:=&\sum_{i=1}^m\int_{\D}(-\beta v_i)\cdot\nabla V(x_i-x_{m+1})\\
&&\times\sum_{k=0}^{m}(-1)^{m-k}\big(\tfrac{N+1-k}N\big)\sum_{\sigma\in\mathfrak P_{k}^{m+1}\atop i\notin\sigma}\tfrac{F_{N}^{k+1}}{M_\beta^{\otimes(k+1)}}(z_0,z_\sigma)\,M_\beta(v_{m+1})\,dz_{m+1}.
\end{eqnarray*}
First, using the cluster expansion for marginals in terms of cumulants (cf.~Lemma~\ref{lem:cumulants}(ii)), we find after straightforward combinatorial computations,
\begin{eqnarray*}
\lefteqn{\sum_{k=0}^{m+1}(-1)^{m+1-k}\big(\tfrac{N+1-k}{N}\big)\sum_{\sigma\in\mathfrak P_{k}^{m+1}}\tfrac{F_{N}^{k+1}}{M_\beta^{\otimes(k+1)}}(z_0,z_\sigma)}\\
&=&\sum_{k=0}^{m+1}(-1)^{m+1-k}\big(\tfrac{N+1-k}{N}\big)\sum_{\sigma\in\mathfrak P_{k}^{m+1}} \sum_{j=0}^k\sum_{\tau\in\mathfrak P_j^\sigma}G_N^{j+1}(z_0,z_\tau)\\
&=&\sum_{j=0}^{m+1}\bigg(\sum_{k=j}^{m+1}(-1)^{m+1-k}\big(\tfrac{N+1-k}{N}\big)\binom{m+1-j}{k-j}\bigg)\sum_{\tau\in\mathfrak P_j^{m+1}}G_N^{j+1}(z_0,z_\tau)\\
&=&\tfrac{N-m}N\,G_N^{m+2}(z_0,z_{[m+1]})-\tfrac1N\sum_{\tau\in\mathfrak P_m^{m+1}}G_N^{m+1}(z_0,z_\tau),
\end{eqnarray*}
which leads to the identity
\begin{multline*}
B_0=\tfrac{N-m}N\sum_{i=0}^m\int_{\D}\nabla V(x_i-x_{m+1})\cdot(\nabla_{v_i}-\beta v_i) G_N^{m+2}(z_0,z_{[m+1]})\,M_\beta(v_{m+1})\,dz_{m+1}\\
-\tfrac1N\sum_{i=0}^m\sum_{j=1}^m\int_{\D}\nabla V(x_i-x_{m+1})\cdot(\nabla_{v_i}-\beta v_i) G_N^{m+1}(z_0,z_{[m+1]\setminus\{j\}})\,M_\beta(v_{m+1})\,dz_{m+1}.
\end{multline*}
Second, for $1\le i\le m$, we similarly compute
\begin{multline*}
\sum_{k=0}^{m}(-1)^{m-k}\big(\tfrac{N+1-k}{N}\big)\sum_{\sigma\in\mathfrak P_{k}^{m+1}\atop i\notin\sigma}\tfrac{F_{N}^{k+1}}{M_\beta^{\otimes(k+1)}}(z_0,z_\sigma)\\
\,=\,\tfrac{N+1-m}N\,G_N^{m+1}(z_0,z_{[m+1]\setminus\{i\}})
-\tfrac1N\sum_{\tau\in\mathfrak P^{m+1}_{m-1}\atop i\notin\tau}G_N^m(z_0,z_{\tau}),
\end{multline*}
which leads to
\begin{multline*}
B_1=\tfrac{N+1-m}N\sum_{i=1}^m\int_{\D}\nabla V(x_i-x_{m+1})\cdot(-\beta v_i)\,G_N^{m+1}(z_0,z_{[m+1]\setminus\{i\}})\,M_\beta(v_{m+1})\,dz_{m+1}\\
-\tfrac1N\sum_{1\le i,j\le m\atop i\ne j}\int_{\D}\nabla V(x_i-x_{m+1})\cdot(-\beta v_i)\,G_N^m(z_0,z_{[m+1]\setminus\{i,j\}})\,M_\beta(v_{m+1})\,dz_{m+1}.
\end{multline*}
The conclusion follows.
\end{proof}

\section{Formal argument and main difficulties}\label{sec:echo}

In this section, we describe at a formal level the string of arguments that lead to the expected Fokker-Planck equation~\eqref{eq:FP} for the tagged particle velocity density, and we emphasize the main difficulties that arise.
Since the mean-field force exerted by the equilibrium bath vanishes, the phase-space probability density $F_N^1=M_\beta G_N^1$ of the tagged particle satisfies the following equation (cf.~Lemma~\ref{lem:eqns}),
\[(\partial_t+v\cdot\nabla_x)G_N^1=M_2^1G_N^2:=(\nabla_v-\beta v)\cdot\int_\D\nabla V(x-x_*)G_N^2(z,z_*)\,M_\beta(v_*)\,dz_*,\]
where the correction to pure transport is thus dictated by the correlation $G_N^2$ of the tagged particle with the background.
As we expect this correction to play a role on long timescale, it is natural to filter out the oscillations created by the spatial transport on smaller timescales. We will therefore
focus on the projection onto the kernel of the  transport, that is, on the velocity distribution, 
\begin{gather*}
f_N^1(v):=\int_{\T^d}F_N^1(x,v)\,dx,\qquad g_N^1(v):=\frac{f_N^1}{M_\beta}(v)=\int_{\T^d}G_N^1(x,v)\,dx,\\
\partial_tg_N^1=\int_{\T^d}(M_2^1G_N^2)(x,\cdot)\,dx.
\end{gather*}
As we expect $G_N^2=O(\frac1N)$, we may either look for the correction $O(\frac1N)$ on the mean-field timescale $t\sim1$, or describe the leading behavior on the relevant long timescale $t\sim N$. As emphasized below, although the two questions are often confused in the physics literature, they are mathematically not equivalent.
On the one hand, characterizing corrections of order $O(\frac1N)$ on the mean-field timescale requires a refined multiscale analysis and a strong control on error terms. In particular, oscillations need to be described precisely. In this respect, looking at the long timescale $t\sim N$ may seem easier as the effect of correlations becomes macroscopic and can be caught with a weaker notion of convergence. The drawback is that we then need to establish long-time a priori estimates and make sure that they are strong enough to control resonances.

\subsection{Formal analysis of correlations}
We proceed to a formal examination of the BBGKY hierarchy to capture the effects of correlations on the tagged particle dynamics. A quick observation of the equations for successive cumulants in Lemma~\ref{lem:eqns}, starting with the particular initial data~\eqref{eq:tagged-initial},
leads us to actually expect the following optimal decay of correlations,
\begin{equation}\label{eq:orders-cum}
G_N^1=O(1),\qquad G_N^{2k},\,G_N^{2k+1}=O\big(\tfrac1{N^k}\big)\quad\text{for $k\ge1$}.
\end{equation}
Starting with the suboptimal a priori decay stated in Lemma~\ref{lem:cum-ap} and using the equations for cumulants, we can easily prove that these orders of magnitude indeed hold true in weak topology on the mean-field timescale $t\sim1$ as a particular case of~Lemma~\ref{lem:truncate-eqn-G123} below.
Surprisingly, the above scaling~\eqref{eq:orders-cum} entails that the second and third cumulants $G_N^2$ and $G_N^3$ are expected to be of the same order $O(\frac1N)$: as opposed to the nonlinear setting described in Section~\ref{sec:overview} where it has higher order and can be neglected (cf.~\cite{D-19}), the three-particle correlation $G_N^3$ can thus no longer be neglected here and indeed leads to a nontrivial contribution of the same order as $G_N^2$. This is due to the fact that linear cumulants are conveniently defined in Lemma~\ref{lem:cumulants} by linearizing at the Maxwellian distribution $M_\beta^{\otimes N}$ without taking into account the spatial correlations of the correct Gibbs measure $M_{N,\beta}$ initially describing the equilibrium background. 
Due to this error, the corresponding BBGKY hierarchy of equations for cumulants in Lemma~\ref{lem:eqns} can only be truncated at~$G_N^4$ rather than $G_N^3$, and it takes the form of the following Bogolyubov type equations~\cite{Bogolyubov-46},
\begin{eqnarray}
(\partial_t+iL_{1})G_N^1&=&\tfrac1NM^1_2(NG_N^2),\nonumber\\
(\partial_t+iL_{2})(NG_N^2)&=&S_1^2G_N^1+M^2_3(NG_N^3)+O(\tfrac1{N}),\nonumber\\
(\partial_t+iL_{3})(NG_N^3)&=&S_1^3G_N^1+O(\tfrac1{N}),\label{eq:truncate-BBGKY}
\end{eqnarray}
which can be solved by means of Duhamel's formula, with $G_N^{2;\circ}=0$,
\begin{multline}\label{eq:nonMarkov}
(\partial_t+iL_1)G_N^1=\tfrac1N\int_0^tM_2^1e^{-iL_2(t-t')}S_1^2G_N^{1;t'}\,dt'\\
+\tfrac1N\int_0^t\!\!\int_0^{t'}M_2^1e^{-iL_2(t-t')}M_3^2e^{-iL_3(t'-t'')}S_1^3G_N^{1;t''}\,dt''dt'\\
+\tfrac1N\Big(\int_0^tM_2^1e^{-iL_2(t-t')}M_3^2e^{-iL_3t'}dt'\Big)(NG_N^{3;\circ})+O(\tfrac1{N^2}),
\end{multline}
where the spatial transport operator $iL_1:=v\cdot\nabla_x$ drops when turning to the velocity distribution $g_N^1$.
As shown in Proposition~\ref{prop:truncation} below, this first-order correction to pure transport can be rigorously justified on the mean-field timescale $t\sim1$. Viewed as an equation on $g_N^1$, this is however not very illuminating as the correction takes a complicated non-Markovian form. Complicated memory effects are typically expected to become negligible on long timescales $t\gg1$ (e.g.~\cite{Spohn-80}), and the above equation on $g_N^1$ would take the simpler form of a diffusive equation on the relevant timescale $t\sim N$, showing in particular that the tagged particle trajectory becomes Brownian under the effect of the background.
By means of complex deformation computations, we show in Section~\ref{sec:eqns-lim} below that the long-time limit of the $O(\frac1N)$ correction in~\eqref{eq:nonMarkov} precisely leads to the predicted Fokker-Planck operator~\eqref{eq:FP}. To complete the argument, it remains to justify that~\eqref{eq:nonMarkov} also holds with essentially the same error estimate uniformly on longer timescales $t\gg1$.

\subsection{Justifying the Markovian limit: a question of timescales}
While the problem is reduced as above to propagating~\eqref{eq:nonMarkov} on longer timescales $t\gg1$, two main difficulties show up:
\begin{enumerate}[$\bullet$]
\item
The expected decay $G_N^{2k},G_N^{2k+1}=O(\frac1{N^k})$, cf.~\eqref{eq:orders-cum}, can only be proved on the mean-field timescale $t\sim1$. On longer timescales, rather appealing to the suboptimal time-uniform estimate $G_N^{m+1}=O(\frac1{N^{m/2}})$, cf.~Lemma~\ref{lem:cum-ap}, additional spurious terms involving $G_N^2,G_N^4$ have therefore to be taken into account in~\eqref{eq:nonMarkov}.
\item The uniform estimate $G_N^{m+1}=O(\frac1{N^{m/2}})$ holds only in weak spaces, typically in~$\Ld^2_\beta$, cf.~Lemma~\ref{lem:cum-ap}.
Since the operators $M$'s and $S$'s in Lemma~\ref{lem:eqns} involve $v$-derivatives, the error terms in the truncated equations~\eqref{eq:truncate-BBGKY} involve $v$-derivatives of higher-order cumulants, and the resulting lack of regularity of the errors may come in resonance with the singularity of the long-time propagators for the linearized Vlasov operators $L$'s, leading to possibly diverging contributions.
\end{enumerate}
In order to illustrate the difficulty, let us consider a typical error term from~\eqref{eq:truncate-BBGKY}, that is,
\[R^t_N:=\tfrac1{N^2}\int_0^tM_2^1e^{-iL_2(t-t')}S_2^2(NG_N^{2;t'})\,dt',\]
which we aim to analyze on the long timescale $t=N^r\tau$, $\tau\sim1$, for some $r>0$.
This is equivalently reformulated in terms of Laplace transform: with the notation $\widetilde R^\alpha_N:=\int_0^\infty R^{N^r\tau}_Ne^{-\alpha\tau}d\tau$, $\Re\alpha>0$, we compute
\begin{equation}\label{eq:def-R-toest}
\widetilde R^\alpha_N= \tfrac1{N^2}M_2^1(iL_2+\tfrac\alpha{N^r})^{-1}S_2^2(N\widetilde G_N^{2;\alpha}).
\end{equation}
In order to be negligible as desired in~\eqref{eq:nonMarkov}, this term needs to be shown of order $o(\frac1N)$.
Assuming for simplicity that the linearized Vlasov operator $iL_2$ is replaced by the pure transport operator $iL_2':=v_0\cdot\nabla_{x_0}+v_1\cdot\nabla_{x_1}$,
the long-time propagator is explicitly checked to display violent oscillations in phase space, which is known as filamentation: indeed, the evolution
\begin{equation}\label{eq:filament}
e^{-iL_2'N^r\tau}h(z_0,z_1)=h\big((x_0-N^r\tau v_0,v_0),(x_1-N^r\tau v_1,v_1)\big)
\end{equation}
exhibits oscillations at the scale $O(\frac1{N^r})$ and converges in the long-time limit $t\sim N^r\uparrow\infty$ towards the average in position but only in the weak sense.
Correspondingly, the resolvent $(iL_2'+\frac\alpha{N^r})^{-1}$ takes on the following form in Fourier space,
\begin{multline}\label{eq:longtime-prop}
(i\widehat L_2'+\tfrac\alpha{N^r})^{-1}\widehat G\,=\,
\big(ik_0\cdot v_0+ik_1\cdot v_1+\tfrac\alpha{N^r}\big)^{-1}\widehat G\\
~\xrightarrow{N\uparrow\infty}~\Big(\pi\delta(k_0\cdot v_0+k_1\cdot v_1)-i\pv\tfrac1{k_0\cdot v_0+k_1\cdot v_1}\Big)\widehat G,
\end{multline}
where by the Sokhotski-Plemelj formula the distributional limit does not correspond to a bounded operator on $\Ld^2_\beta(\D^2)$.
Provided that $G_N^2$ is of order $o(1)$ in the smooth topology on the timescale $t\sim N^r$, we would deduce that the error $R_N$ has negligible size $o(\frac1N)$ when averaged with respect to velocity against a smooth test function.
However, the only a priori estimates at hand for $G_N^2$ hold in $\Ld^2_\beta(\D^2)$, cf.~Lemma~\ref{lem:cum-ap}: even if $G_N^2$ was known to be of the optimal order $O(\frac1N)$ in that space uniformly in time,
since the operator $S_2^2$ contains a $v$-derivative, we easily check that $R_N$ would a priori be at best of order $O(\frac{N^{3r/2}}{N^2})$ when averaged with a smooth test function.
On the intermediate timescale $t\sim N^r$ with $r<\frac23$, $R_N$ is thus of order $o(\frac1N)$, hence negligible as desired.
In contrast, on longer timescales, more information is required on possible singularities of $G_N^2$ to ensure that those do not create resonances with oscillations of the long-time propagator.
Suitably unravelling such finer information is left as an open question.

\subsection{A full series expansion}
We provide a formal argument suggesting that violent oscillations of the long-time propagators in phase space never come in resonance with the corresponding oscillations of cumulants.
For that purpose, instead of aiming to truncate the BBGKY hierarchy as in~\eqref{eq:nonMarkov}, we rather iteratively replace higher-order cumulants in terms of the Duhamel formula for the full corresponding equations as given in Lemma~\ref{lem:eqns}. This leads to the following closed equation for $G_N^1$ in form of an infinite series expansion,
\begin{multline}\label{eq:formal-g1}
(\partial_t+iL_1)G_N^1\\
=\sum_{n=1}^\infty\tfrac1{N^{n}}\sum_{\ell=n+1}^{3n}\int_{(\R^+)^{\ell}}\delta\Big(t-\sum_{j=1}^\ell t_j\Big)\bigg(\sum_{\sigma\in J_\ell^n}A^{\sigma_1}e^{-it_1L}A^{\sigma_2}\ldots A^{\sigma_{\ell-1}}e^{-it_{\ell-1} L}A^{\sigma_\ell}G_N^{1;t_\ell}\\
+\sum_{\sigma\in K_\ell^n}A^{\sigma_1}e^{-it_1L}A^{\sigma_2}\ldots A^{\sigma_\ell}e^{-it_\ell L}\big(N^{\sigma_1^\ell-1}G_N^{\sigma_1^\ell+1;\circ}\big)\bigg)dt_1\ldots dt_\ell,
\end{multline}
where the index sets $J_\ell^n$'s and $K_\ell^n$'s are given by
\begin{eqnarray*}
J_\ell^n&:=&\Big\{\sigma\in\{-2,-1,0,1\}^\ell:\,\sigma_1^\ell=0,~\sigma_1^j>0~\forall1\le j<\ell,~\gamma(\sigma)=n\Big\},\\
K_\ell^n&:=&\Big\{\sigma\in\{-2,-1,0,1\}^\ell:\,\sigma_1^\ell>1,~\sigma_1^j>0~\forall1\le j\le\ell,~\gamma(\sigma)+\sigma_1^\ell=n+1\Big\},
\end{eqnarray*}
with the shorthand notation $\sigma_1^j:=1+\sum_{l=1}^j\sigma_l$ and $\gamma(\sigma):=\#\{j:\sigma_j\le0\}$,
and where the operators $A$'s and $L$'s are given by
\[A^{-2}\in\{S^{m+2}_m\}_m,\quad A^{-1}\in\{S^{m+1}_m\}_m,\quad A^0\in\{S^m_m\}_m,\quad A^{1}\in\{M^{m-1}_m\}_m,\quad L\in\{L_m\}_m,\]
while for notational simplicity we omit subscripts $m$'s, which are indeed anyway uniquely determined for each term.
Note that elements in the index sets $J_\ell^n$'s and $K_\ell^n$'s can be viewed as subsets of walks with steps $-2$, $-1$, $0$, or $1$, starting at site $0$. The terms of formal order $O(\frac1N)$ obtained for $n=1$ in~\eqref{eq:formal-g1} are checked to coincide with the right-hand side terms in~\eqref{eq:nonMarkov}.
If $G_N^1$ is controlled in the smooth topology and if for all $m\ge2$ the initial data $G_N^{m+1;\circ}$ are similarly controlled by $O(\frac1{N^{m-1}})$, then each term in the series~\eqref{eq:formal-g1} can be checked to admit a well-defined long-time limit $t\uparrow\infty$ when tested with a smooth test function. 
While the direct analysis of error terms such as~\eqref{eq:def-R-toest} was hindered by possible resonances between $G_N^2$ and oscillations of the long-time propagator, the present full expansion suggests that such resonances should not create diverging contributions, at least term by term.
However, although the index sets $J_\ell^n$'s and $K_\ell^n$'s have moderate (exponential) size, the following two issues prevent the summability of the series in the long-time limit and obstruct any rigorous analysis:
\begin{enumerate}[$\bullet$]
\item Each occurrence of an operator $A$ contains a $v$-derivative. Summability of the series then requires to restrict to an analytic setting.
\item The $v$-derivatives stemming from the operators $A$'s do not commute with the linearized propagators $e^{-itL}$'s, cf.~\eqref{eq:filament}. In order to compute the long-time limit in each term, we must then first proceed to multiple integrations by parts, which generates a factorial number of terms that do not seem to recombine nicely in the limit.
\end{enumerate}
For those reasons, we do not know how to make this perturbative approach rigorous. Still, it suggests that a compensation mechanism is hidden in the BBGKY hierarchy, which would systematically avoid divergences.
A diagrammatic approach to take advantage of such compensations has been proposed by Prigogine and Balescu~\cite{PB-59a,PB-59b}, but the loss of derivatives and the convergence issues are not addressed.

\medskip
In the next sections, we rather focus on intermediate timescales $t\sim N^r$ with $r>0$ small enough, for which propagators have tamer oscillations.

\section{Truncation of the BBGKY hierarchy on intermediate timescale}\label{sec:intermediate}

This section is devoted to the rigorous truncation of the BBGKY hierarchy for cumulants on intermediate timescales $t\sim N^r$ with $r>0$ small enough. More precisely, on such timescales, we prove the following rigorous version of the integrated form~\eqref{eq:nonMarkov} of the non-Markovian Bogolyubov equations. Note that for $t\ll N$ all occurrences of $G_N^1$ in the right-hand side of~\eqref{eq:nonMarkov} can indeed be replaced by the initial data~$g^\circ$ to leading order.

\begin{prop}\label{prop:truncation}
Let the assumptions of Theorem~\ref{th:main} hold for $V,\beta,f^\circ$.
Given $0\le r<\frac1{18}$, there holds for all $\tau\ge0$,
\begin{align*}
&\bigg\|N(\partial_t g_N^{1})|_{t=N^r\tau}
-\int_0^{N^r\tau} \int_{\T^d}\big(M_2^1e^{-iL_2(N^r\tau-t')} S_1^2g^\circ\big)(x,\cdot)\,dx\,dt'\\
&\qquad-\!\!\int_0^{N^r\tau}\!\!\!\int_0^{t'} \int_{\T^d}\big(M_2^1e^{-iL_2(N^r\tau-t')}M_3^2 e^{-iL_3(t'-t'')}S_1^3g^\circ\big)(x,\cdot)\,dx\,dt''dt'\\
&\qquad
-\int_0^{N^r\tau}\int_{\T^d} \big(M_2^1e^{-iL_2(N^r\tau-t')}M_3^2e^{-iL_3t'}(NG_N^{3;\circ})\big)(x,\cdot)\,dx\,dt'
\bigg\|_{H^{-4}_\beta(\R^d)}\\
&\hspace{11cm}\,\lesssim_\beta\,\langle\tau\rangle^9N^{9r-\frac12}.
\qedhere
\end{align*}
\end{prop}

While the above is quantified in $H^{-4}_\beta(\R^d)$, we define the whole string of Sobolev spaces with respect to the Maxwellian distribution.
For $s\ge0$, we first define $H^s_\beta(\D^{m+1})$ as the Hilbert subspace of $\Ld^2_\beta(\D^{m+1})$ with the norm
\[\|G\|_{H^s_\beta(\D^{m+1})}^2=\int_{\D^{m+1}}|\langle(\nabla_x,\nabla_v)\rangle^sG|^2M_\beta^{\otimes(m+1)},\]
where we use the shorthand notation $(x,v)=(x_{[0,m]},v_{[0,m]})\in\D^{m+1}$ and where we recall the notation $\langle\nabla\rangle=(1-\triangle)^{1/2}$.
We denote by $H^{-s}_\beta(\D^{m+1})$ the dual of $H^s_\beta(\D^{m+1})$ with respect to the scalar product of $\Ld^2_\beta(\D^{m+1})$, which is again a Hilbert space with the dual norm
\[\|G\|_{H^{-s}_\beta(\D^{m+1})}^2=\int_{\D^{m+1}}|\langle(\nabla_x,\nabla_v-\beta v)\rangle^{-s}G|^2M_\beta^{\otimes(m+1)}.\]
Note that $\nabla_v-\beta v$ and $\nabla_v$ are actually equivalent in the definition of these weighted spaces.
We similarly write $H^s_\beta((\R^d)^{m+1})$ for the subspace of $H^s_\beta(\D^{m+1})$ of functions that do not depend on $x$.

\subsection{Linearized Vlasov operators}
We start with a spectral description of the linearized Vlasov operator $L_{m+1}$ as defined in Lemma~\ref{lem:eqns}.
We include an explicit computation of the resolvent, although this is only needed in Section~\ref{sec:eqns-lim}.
The positivity property for the (approximate) dispersion function $\e_{\beta,m}^\circ$ in~(ii) is related to the stability of the Maxwellian equilibrium (e.g.~\cite[Section~9.2]{Krall-Trivelpiece}) and is further refined in Lemma~\ref{lem:prel-est}.
Note that a more general situation is discussed in~\cite[Sections~1.1--1.2]{Degond-86}, where linearization is performed around a non-Maxwellian equilibrium.

\begin{lem}[Linearized Vlasov operator]\label{lem:spectrum}
Given $V\in W^{1,\infty}(\T^d)$ with $\widehat V\ge0$, the following hold for all $0\le j\le m\le N$ and $\beta\in(0,\infty)$,
\begin{enumerate}[(i)]
\item The operator $L_{m+1}^{(j)}$ is a bounded perturbation of a self-adjoint operator, hence it generates a $C_0$-group $\{e^{itL_{m+1}^{(j)}}\}_{t\in\R}$.
\item The resolvent of $L_{m+1}^{(j)}$ takes on the following explicit form, for $\omega\in\C\setminus\R$ and $j\ne0$, in Fourier space,
\begin{gather*}
\qquad(\widehat L_{m+1}^{(j)}-\omega)^{-1}\widehat G\,=\,\tfrac{\widehat G}{k_j\cdot v_j-\omega}-\tfrac{N+1-m}N\tfrac{\beta\widehat V(k_j)}{\e_{\beta,m}^\circ(k_j,\omega)}\tfrac{k_j\cdot v_j}{k_j\cdot v_j-\omega}\big\langle\tfrac{\widehat G}{k_j\cdot v_j-\omega}\big\rangle_{v_j},\\
\qquad\big\langle(\widehat L_{m+1}^{(j)}-\omega)^{-1}\widehat G\big\rangle_{v_j}=\tfrac{1}{\e_{\beta,m}^{\circ}(k_j,\omega)}\big\langle\tfrac{\widehat G}{k_j\cdot v_j-\omega}\big\rangle_{v_j},
\end{gather*}
in terms
of
\[\qquad\e_{\beta,m}^{\circ}(k,\omega)\,:=\,1+\tfrac{N+1-m}N\beta\widehat V(k)\big\langle\tfrac{k\cdot v}{k\cdot v-\omega}\big\rangle_{v} ,\]
which satisfies for all $\omega\in\C\setminus\R$ and $k\in\Z^d$,
\[\qquad|\e_{\beta,m}^\circ(k,\omega)|\,\ge\,1-\tfrac{|\Re\omega|}{|\omega|}\,>\,0.\]
\item The spectrum of $L_{m+1}^{(j)}$ coincides with $\R$, and is absolutely continuous with an eigenvalue embedded at $0$ (with eigenspace $\{\psi\in\Ld^2_\beta(\D^{m+1}):\nabla_{x_j}\psi\equiv0\}$).
\item The $C_0$-group $\{e^{itL_{m+1}^{(j)}}\}_{t\in\R}$ is uniformly bounded, that is,
\[\qquad\sup_{t\in\R}\|e^{itL_{m+1}^{(j)}}G\|_{\Ld^2_\beta(\D^{m+1})}\,\lesssim\,\|G\|_{\Ld^2_\beta(\D^{m+1})}\qquad\text{for all $G\in\Ld^2_\beta(\D^{m+1})$}.\]
\end{enumerate}
In addition, the operator $L_{m+1}=\sum_{j=0}^mL_{m+1}^{(j)}$ is the sum of $m+1$ commuting operators, hence it also generates a uniformly bounded $C_0$-group.
\end{lem}

\begin{proof}
Let $0\le j\le m$ be fixed. We start with the proof of~(i). The essentially self-adjoint operator $L_{m+1}^{(j),0}:=v_j\cdot\frac1i\nabla_{x_j}$ generates a $C_0$-group of isometries on $\Ld^2_\beta(\D^{m+1})$, which is explicitly given by
\[e^{itL_{m+1}^{(j),0}}G(z_{[0,m]})=e^{tv_j\cdot\nabla_{x_j}}G(z_{[0,m]})=G\big(z_{[0,j-1]},(x_j+tv_j,v_j),z_{[j+1,m]}\big).\]
Since the operator $L_{m+1}^{(j)}$ takes the form $L_{m+1}^{(j)}=L_{m+1}^{(j),0}+L_{m+1}^{(j),1}$, where the perturbation $L_{m+1}^{(j),1}$ is obviously bounded on $\Ld^2_\beta(\D^{m+1})$, it follows from standard perturbation theory (e.g.~\cite[Theorem~IX.2.1]{Kato}) that the perturbed operator $L_{m+1}^{(j)}$ also generates a $C_0$-group on $\Ld^2_\beta(\D^{m+1})$.

\medskip\noindent
We turn to the resolvent computation~(ii).
Letting $\omega\in\C\setminus\R$, $j\ne0$, and $G\in\Ld^2_\beta(\D^{m+1})$, we aim to compute $H_\omega:=(L_{m+1}^{(j)}-\omega)^{-1}G$. By definition of $L_{m+1}^{(j)}$ in Lemma~\ref{lem:eqns}, the identity $(L_{m+1}^{(j)}-\omega)H_\omega=G$ takes on the following form in Fourier variables,
\[\widehat G=(k_j\cdot v_j-\omega)\widehat H_\omega+\tfrac{N+1-m}N\beta\widehat V(k_j) k_j\cdot v_j\langle \widehat H_\omega\rangle_{v_j},\]
or equivalently, dividing by $k_j\cdot v_j-\omega$,
\begin{equation}\label{eq:comput-resolv0}
\tfrac{\widehat G}{k_j\cdot v_j-\omega}=\widehat H_\omega+\tfrac{N+1-m}N\beta\widehat V(k_j)\tfrac{ k_j\cdot v_j}{k_j\cdot v_j-\omega}\langle \widehat H_\omega\rangle_{v_j}.
\end{equation}
Averaging with respect to $v_j$ yields
\[\big\langle\tfrac{\widehat G}{k_j\cdot v_j-\omega}\big\rangle_{v_j}=\e_{\beta,m}^\circ(k_j,\omega)\langle \widehat H_\omega\rangle_{v_j},\]
with $\e_{\beta,m}^\circ$ defined in the statement.
Since $\widehat V$ is real-valued by assumption, we compute
\begin{multline}\label{eq:bound-eps-pre}
|\e_{\beta,m}^\circ(k,\omega)|^2\,=\,\Big|1+\tfrac{N+1-m}N\beta\widehat V(k)\big\langle\tfrac{ k\cdot v}{k\cdot v-\omega}\big\rangle_{v} \Big|^2\\
\,=\,\Big(1+\tfrac{N+1-m}N\beta\widehat V(k)\big\langle\tfrac{(k\cdot v-\Re\omega)^2}{(k\cdot v-\Re\omega)^2+(\Im\omega)^2}\big\rangle_{v}
+(\Re\omega)\tfrac{N+1-m}N\beta\widehat V(k)\big\langle\tfrac{k\cdot v-\Re\omega}{(k\cdot v-\Re\omega)^2+(\Im\omega)^2}\big\rangle_{v} \Big)^2\\
+\Big((\Im\omega)\tfrac{N+1-m}N\beta\widehat V(k)\big\langle\tfrac{k\cdot v}{(k\cdot v-\Re\omega)^2+(\Im\omega)^2}\big\rangle_{v}\Big)^2,
\end{multline}
and hence, in view of the inequality
\[(a+b\Re \omega)^2+(b\Im \omega)^2\ge(a^2+b^2|\omega|^2)\big(1-\tfrac{|\Re \omega|}{|\omega|}\big)\qquad\text{for all~$a,b\in\R$,}\]
and recalling that $\widehat V$ is nonnegative by assumption, we deduce
\begin{align*}
|\e_{\beta,m}^\circ(k,\omega)|^2\,\ge\,1-\tfrac{|\Re \omega|}{|\omega|}>0.
\end{align*}
As in particular $\e_{\beta,m}^\circ$ does not vanish, the above becomes
\[\langle\widehat H_\omega\rangle_{v_j}=\tfrac1{\e_{\beta,m}^\circ(k_j,\omega)}\big\langle\tfrac{\widehat G}{k_j\cdot v_j-\omega}\big\rangle_{v_j},\]
which yields, once combined with~\eqref{eq:comput-resolv0},
\begin{equation*}
\widehat H_\omega\,=\,\tfrac{\widehat G}{k_j\cdot v_j-\omega}-\tfrac{N+1-m}N\tfrac{\beta \widehat V(k_j)}{\e^\circ_{\beta,m}(k_j,\omega)}\tfrac{ k_j\cdot v_j}{k_j\cdot v_j-\omega}\big\langle\tfrac{\widehat G}{k_j\cdot v_j-\omega}\big\rangle_{v_j}.
\end{equation*}
Examining carefully this resolvent computation leads to the characterization~(iii) of the spectrum. 

\medskip\noindent
Finally, property~(iv) follows from the general characterization of the growth bound of $C_0$-groups in e.g.~\cite[Corollary~A-III.7.11]{Nagel-86}, together with the above resolvent computation.
We also provide a more direct argument based on energy conservation.
Define the following modified scalar product on~$\Ld^2_\beta(\D^{m+1})$ (recall that $\widehat V\ge0$),
\begin{multline*}
\langle H,G\rangle_{\widetilde\Ld^2_{\beta,j}(\D^{m+1})}\,:=\,\langle H,G\rangle_{\Ld^2_{\beta}(\D^{m+1})}\\
+\tfrac{N+1-m}N\int_{\D^{m+2}}V(x_j-x_*)\,\overline{H(z_{[0,m]})}\,G(z_{[0,m]\setminus\{j\}},z_*)\,M_\beta^{\otimes(m+2)}(z_{[0,m]},z_*)\,dz_{[0,m]}dz_*,
\end{multline*}
note that the corresponding norm is Lipschitz-equivalent to the norm of $\Ld^2_\beta(\D^{m+1})$,
\[\|G\|_{\Ld^2_{\beta}(\D^{m+1})}\,\le\,\|G\|_{\widetilde\Ld^2_{\beta,j}(\D^{m+1})}\,\le\,\big(1+\|V\|_{\Ld^\infty(\T^d)}\big)^\frac12\|G\|_{\Ld^2_{\beta}(\D^{m+1})},\]
and denote by $\widetilde\Ld^2_{\beta,j}(\D^{m+1})$ the Hilbert space $\Ld^2_{\beta}(\D^{m+1})$ endowed with this new structure.
Since a straightforward computation shows that $L_{m+1}^{(j)}$ is essentially self-adjoint on~$\widetilde\Ld^2_{\beta,j}(\D^{m+1})$, Stone's theorem yields
\[\|e^{itL_{m+1}^{(j)}}G\|_{\Ld^2_\beta(\D^{m+1})}\,\le\,\|e^{itL_{m+1}^{(j)}}G\|_{\widetilde\Ld^2_{\beta,j}(\D^{m+1})}\,=\,\|G\|_{\widetilde\Ld^2_{\beta,j}(\D^{m+1})}\,\lesssim\,\|G\|_{\Ld^2_\beta(\D^{m+1})},\]
and the conclusion follows.
\end{proof}

Next, we establish the following estimate for the $C_0$-group generated by the linearized Vlasov operator $L_{m+1}$ in weak norms.

\begin{lem}[Weak bounds on linearized Vlasov evolution]\label{lem:prop-L}
Given $V\in W^{1,\infty}(\T^d)$ with $\widehat V\ge0$, there holds for all $\beta\in(0,\infty)$, $0\le m\le N$, $s\ge0$, and $G\in C^\infty_c(\D^{m+1})$,
\[\|e^{iL_{m+1}t}G\|_{H^{-s}_\beta(\D^{m+1})}\,\lesssim_{\beta,m,s}\,\langle t\rangle^s\|G\|_{H^{-s}_\beta(\D^{m+1})}.\qedhere\]
\end{lem}

\begin{proof}
By duality, it suffices to prove
\begin{equation}\label{eq:prop-reg}
\|e^{-iL_{m+1}^*t}G\|_{H^s_\beta(\D^{m+1})}\,\lesssim_{\beta,m,s}\,\langle t\rangle^s\|G\|_{H^s_\beta(\D^{m+1})},
\end{equation}
where the adjoint operator is given by
\begin{multline*}
iL_{m+1}^*G\,=\,\sum_{j=0}^{m}\Big(v_j\cdot\nabla_{x_j}G+\mathds1_{j\ne0}\tfrac{N+1-m}N\int_\D\beta v_*\cdot\nabla V(x_j-x_*)\\
\times G(z_0,z_{[m]\setminus\{j\}},z_*)\,M_\beta(v_*)\,dz_*\Big).
\end{multline*}
The propagation of regularity~\eqref{eq:prop-reg} relies crucially on the following two observations:
\begin{equation}\label{x-commutator} 
[\nabla_{x_l},iL_{m+1}^*]\,=0\,,
\end{equation}
and 
\begin{equation}\label{eq:element-comm}
[\nabla_{v_l},iL_{m+1}^*]\,=\,\nabla_{x_l}+R_{m+1}^{(l)},
\end{equation}
where $R_{m+1}^{(l)}$ is the bounded operator on $\Ld^2_\beta(\D^{m+1})$ defined by
\begin{equation*}
R_{m+1}^{(l)}G\,=\,\mathds1_{l\ne0}\tfrac{N+1-m}N\int_{\D}(\beta\Id-\beta v_*\otimes\beta v_*)\nabla V(x_l-x_{*}) \,G(z_0,z_{[m]\setminus\{l\}},z_*)\,M_\beta(v_{*})\,dz_{*}.
\end{equation*}
We argue by induction:
since the result~\eqref{eq:prop-reg} with $s=0$ follows from Lemma~\ref{lem:spectrum}(iv), we may assume that it holds for some $s= s_0\ge0$, and we shall deduce that it also holds for $s=s_0+1$.
Since in view of~\eqref{x-commutator} the operator $\nabla_{x_l}$ obviously commutes with the group $\{e^{-iL_{m+1}^*t}\}_t$, it remains to prove
\begin{equation}\label{eq:prop-reg-re}
\|\nabla_{v_l}^{s_0+1}e^{-iL_{m+1}^*t}G\|_{\Ld^2_\beta(\D^{m+1})}\,\lesssim_{\beta,m,s_0}\,\langle t\rangle^{s_0+1}\|G\|_{H^{s_0+1}_\beta(\D^{m+1})}.
\end{equation}
Writing
\[(\partial_t+iL_{m+1}^*)(\nabla_{v_l}^{s_0+1}e^{-iL_{m+1}^*t}G)\,=\,-[\nabla_{v_l}^{s_0+1},iL_{m+1}^*]e^{-iL_{m+1}^*t}G,\]
Duhamel's formula yields 
\[\nabla_{v_l}^{s_0+1}e^{-iL_{m+1}^*t}G-e^{-iL_{m+1}^*t}\nabla_{v_l}^{s_0+1}G\,=\,-\int_0^te^{-iL_{m+1}^*(t-t')}[\nabla_{v_l}^{s_0+1},iL_{m+1}^*]e^{-iL_{m+1}^*t'}G\,dt',\]
and we deduce in view of Lemma~\ref{lem:spectrum}(iv),
\begin{multline*}
\|\nabla_{v_l}^{s_0+1}e^{-iL_{m+1}^*t}G\|_{\Ld^2_\beta(\D^{m+1})}\,\lesssim\,\|\nabla_{v_l}^{s_0+1}G\|_{\Ld^2_\beta(\D^{m+1})}\\
+\int_0^t\|[\nabla_{v_l}^{s_0+1},iL_{m+1}^*]e^{-iL_{m+1}^*t'}G\|_{\Ld^2_\beta(\D^{m+1})}\,dt'.
\end{multline*}
Expanding the commutator as
\[[\nabla_{v_l}^{s_0+1},iL_{m+1}^*]=\sum_{j=0}^{s_0}\nabla_{v_l}^{s_0-j}[\nabla_{v_l},iL_{m+1}^*]\nabla_{v_l}^{j},\]
appealing to~\eqref{eq:element-comm},
using that $\nabla_{x_l}$ commutes with $\nabla_{v_l}$ and with the group $\{e^{-iL_{m+1}^*t}\}_t$,
and noting that for all $s\ge0$ the operator $R_{m+1}^{(l)}$ is bounded by $C\beta$ on $H^s_\beta(\D^{m+1})$ (cf.~$\|(\beta v)^{\otimes2}\|_{\Ld^2_\beta(\D)}\lesssim\beta$),
we are led to
\begin{multline*}
\|\nabla_{v_l}^{s_0+1}e^{-iL_{m+1}^*t}G\|_{\Ld^2_\beta(\D^{m+1})}\,\lesssim_{s_0}\,\|G\|_{H^{s_0+1}_\beta(\D^{m+1})}\\
+\int_0^t\|e^{-iL_{m+1}^*t'}\nabla_{x_l}G\|_{H^{s_0}_\beta(\D^{m+1})}\,dt'
+\beta\int_0^t\|e^{-iL_{m+1}^*t'}G\|_{H^{s_0}_\beta(\D^{m+1})}\,dt'.
\end{multline*}
After combination with the induction assumption (cf.~\eqref{eq:prop-reg} with $s=s_0$), the conclusion~\eqref{eq:prop-reg-re} follows.
\end{proof}

Finally, we show that the operators $M$'s and $S$'s as defined in Lemma~\ref{lem:eqns} amount to the loss of (at most) one $v$-derivative.

\begin{lem}[Bounds on operators $M$'s and $S$'s]\label{lem:prop-MS}
Given $V\in C^\infty(\T^d)$, there holds for all $\beta\in(0,\infty)$, $0\le m\le N$, $s\ge0$, and $G^{m+r}\in C^\infty_c(\D^{m+r})$ with $r\in\{-1,0,1,2\}$,
\begin{align*}
\|M_{m+2}^{m+1}G^{m+2}\|_{H^{-s-1}_\beta(\D^{m+1})}~&\lesssim_{m,s}~\|G_N^{m+2}\|_{H^{-s}_\beta(\D^{m+2})},\\
\|S_{m-1}^{m+1}G^{m-1}\|_{H^{-s}_\beta(\D^{m+1})}~&\lesssim_{\beta,m,s}~\|G_N^{m-1}\|_{H^{-s}_\beta(\D^{m-1})},\\
\|S_{m}^{m+1}G^{m}\|_{H^{-s-1}_\beta(\D^{m+1})}~&\lesssim_{m,s}~\|G_N^{m}\|_{H^{-s}_\beta(\D^{m})},\\
\|S_{m+1}^{m+1}G^{m+1}\|_{H^{-s-1}_\beta(\D^{m+1})}~&\lesssim_{m,s}~\|G_N^{m+1}\|_{H^{-s}_\beta(\D^{m+1})},
\end{align*}
where the multiplicative constants $C_{\beta,m,s}$ further depend on $\|\nabla V\|_{W^{s,\infty}(\T^d)}$.
\end{lem}

\begin{proof}
By duality, it suffices to show
\begin{align*}
\|(M_{m+2}^{m+1})^*G^{m+1}\|_{H^{s}_\beta(\D^{m+2})}~&\lesssim_{m,s}~\|G_N^{m+1}\|_{H^{s+1}_\beta(\D^{m+1})},\\
\|(S_{m-1}^{m+1})^*G^{m+1}\|_{H^{s}_\beta(\D^{m-1})}~&\lesssim_{m,s}~\sqrt\beta\|G_N^{m+1}\|_{H^{s}_\beta(\D^{m+1})},\\
\|(S_{m}^{m+1})^*G^{m+1}\|_{H^{s}_\beta(\D^{m})}~&\lesssim_{m,s}~\|G_N^{m+1}\|_{H^{s+1}_\beta(\D^{m+1})},\\
\|(S_{m+1}^{m+1})^*G^{m+1}\|_{H^{s}_\beta(\D^{m+1})}~&\lesssim_{m,s}~\|G_N^{m+1}\|_{H^{s+1}_\beta(\D^{m+1})}.
\end{align*}
Integrations by parts lead to the following explicit forms of the adjoint operators,
\begin{eqnarray*}
(M_{m+2}^{m+1})^*G_N^{m+1}&=&-\tfrac{N-m}{N(m+1)}\sum_{j=0}^{m+1}\sum_{1\le l\le m+1\atop j\ne l}\nabla V(x_j-x_l)\cdot \nabla_{v_j}G_N^{m+1}(z_0,z_{[m+1]\setminus\{l\}}),\\
(S_{m-1}^{m+1})^*G_N^{m+1}&=&-\tfrac{m(m-1)}2\iint_{\D\times\D}\nabla V(x_*-x_*')\cdot(\beta{v_*}-\beta{v_*'})\,G_N^{m+1}(z_0,z_{[m-2]},z_*,z_*')\\
&&\hspace{5cm}\times M_\beta(v_*)M_\beta(v_*')\,dz_*dz_*',\\
(S_m^{m+1})^*G_N^{m+1}&=&-m\sum_{j=0}^{m-1}\int_\D\nabla V(x_j-x_*)\cdot (\nabla_{v_j}-\nabla_{v_*})G_N^{m+1}(z_0,z_{[m-1]},z_*)\,M_\beta(v_*)\,dz_*\\
&&\hspace{-2.1cm}+m\sum_{j=1}^{m-1}\iint_{\D\times\D}\nabla V(x_{*}-x_j)\cdot\nabla_{v_*}G_N^{m+1}(z_0,z_{[m-1]\setminus\{j\}},z_*,z_*')\,M_\beta(v_{*})M_\beta(v_{*})\,dz_{*}dz_{*}',\\
(S_{m+1}^{m+1})^*G_N^{m+1}&=&-\sum_{0\le j\ne l\le m}\nabla V(x_j-x_l)\cdot \nabla_{v_j}G_N^{m+1}\\
&&\hspace{-1cm}+\sum_{j=0}^m\sum_{1\le l\le m\atop j\ne l}\nabla V(x_j-x_l)\cdot\int_{\D} \nabla_{v_j} G_N^{m+1}(z_0,z_{[m]\setminus\{l\}},z_*)\,M_\beta(v_{*})\,dz_{*}.
\end{eqnarray*}
Direct estimates together with the bound $\|\beta v\|_{\Ld^2_\beta(\D)}\lesssim\sqrt\beta$ then yield the conclusion.
\end{proof}

Lemma~\ref{lem:prop-L} implies that, for fixed $N$, marginals and cumulants are smooth globally in time. Then, by Lemma  \ref{lem:prop-MS}, all terms in the cumulant equations make sense, although it is not yet clear at this stage whether they are uniformly bounded or not.

\subsection{Truncation of the BBGKY hierarchy}
This section is devoted to the proof of Proposition~\ref{prop:truncation}.
We start by rigorously truncating the BBGKY hierarchy of Lemma~\ref{lem:eqns}, leading to a closed system of equations on $G_N^1,G_N^2,G_N^3$. This is based on the uniform a priori estimates of Lemma~\ref{lem:cum-ap},
and constitutes a rigorous version of~\eqref{eq:truncate-BBGKY}.
In addition, this justifies the optimal orders of magnitude in~\eqref{eq:orders-cum}. Although we focus here on the first three cumulants, the estimates are easily pursued to higher order.

\begin{lem}\label{lem:truncate-eqn-G123}
Let the assumptions of Theorem~\ref{th:main} hold for $V,\beta,f^\circ$.
Given $0\le r<\frac14$, there holds for all $\tau\ge0$,
\begin{equation}\label{eq:G1}
N(\partial_t g_N^1)|_{t=N^r\tau}(v)\,=\,\int_{\T^d}\big(M_2^1NG_N^{2;N^r\tau}\big)(x,v)\,dx,
\end{equation}
and
\begin{align*}
\bigg\|NG_N^{2;N^r\tau}
-\int_0^{N^r\tau} e^{-iL_2(N^r\tau-t')}\big(S_1^2G_N^{1;t'}+M_3^2(NG_N^{3;t'})\big)\,dt'\bigg\|_{H^{-1}_\beta(\D^2)}
&\lesssim_\beta\langle\tau\rangle^2N^{2r-\frac12},\\
\bigg\|NG_N^{3;N^r\tau}-e^{-iN^rL_3\tau}(NG_N^{3;\circ})-\int_0^{N^r\tau} e^{-iL_3(N^r\tau-t')}S_1^3G_N^{1;t'}\,dt'\bigg\|_{H^{-1}_\beta(\D^3)}
&\lesssim_\beta\langle\tau\rangle^2N^{2r-\frac12}.
\end{align*}
In particular, this implies
\begin{equation}\label{eq:opt-bound-order3}
\|G_N^{2;N^r\tau}\|_{H^{-1}_\beta(\D^2)}+\|G_N^{3;N^r\tau}\|_{H^{-1}_\beta(\D^2)}\,\lesssim_\beta\,\langle \tau\rangle^{2}N^{2r-1}.
\qedhere
\end{equation}
\end{lem}

\begin{proof}
Averaging in space the equation for $G_N^1$ in Lemma~\ref{lem:eqns} directly yields~\eqref{eq:G1}.
We turn to the description of $NG_N^2$. The corresponding equation in Lemma~\ref{lem:eqns} takes the form
\[\partial_t(NG_N^2)+iL_2(NG_N^2)=S_{1}^2G_N^1+M_{3}^2(NG_N^3)+\tfrac1NS_{2}^2(NG_N^2),\]
and thus, by Duhamel's formula, with $G_N^{2;\circ}=0$ (cf.~Lemma~\ref{lem:initial}),
\begin{multline*}
NG_N^{2;N^r\tau}=
\int_0^{N^r\tau} e^{-iL_2(N^r\tau-t')}\big(S_1^2G_N^{1;t'}+M_3^2(NG_N^{3;t'})\big)\,dt'\\
+\tfrac1N\int_0^{N^r\tau} e^{-iL_2(N^r\tau-t')}S_{2}^2(NG_N^{2;t'})\,dt'.
\end{multline*}
Successively applying Lemmas~\ref{lem:prop-L} and~\ref{lem:prop-MS}, we deduce
\begin{eqnarray*}
\lefteqn{\bigg\|NG_N^{2;N^r\tau}-\int_0^{N^r\tau} e^{-iL_2(N^r\tau-t')}\big(S_1^2G_N^{1;t'}+M_3^2(NG_N^{3;t'})\big)\,dt'\bigg\|_{H^{-1}_\beta(\D^2)}}\hspace{2cm}\\
&\le&\tfrac1N\int_0^{N^r\tau} \big\|e^{-iL_2(N^r\tau-t')}S_{2}^2(NG_N^{2;t'})\big\|_{H^{-1}_\beta(\D^2)}\,dt'\\
&\lesssim_{\beta}&N^{r-1}\langle\tau\rangle\int_0^{N^r\tau}\|S_{2}^2(NG_N^{2;t'})\|_{H^{-1}_\beta(\D^2)}\,dt'\\
&\lesssim_{\beta}&N^{r-1}\langle\tau\rangle\int_0^{N^r\tau}\|NG_N^{2;t'}\|_{\Ld^2_\beta(\D^2)}\,dt',
\end{eqnarray*}
and the stated estimate then follows from the a priori estimate of Lemma~\ref{lem:cum-ap} in the form $\|NG_N^{2}\|_{\Ld^2_\beta(\D^2)}\lesssim_\beta N^\frac12$.

\medskip\noindent
We turn to the corresponding description of $NG_N^3$.
In view of Duhamel's formula, the equation for $NG_N^3$ in Lemma~\ref{lem:eqns} yields
\begin{multline}\label{eq:G3-pre}
NG_N^{3;N^r\tau}=e^{-iN^rL_3\tau}(NG_N^{3;\circ})+\int_0^{N^r\tau}e^{-iL_3(N^r\tau-t')}S_1^3G_N^{1;t'}\,dt'\\
+\int_0^{N^r\tau} e^{-iL_3(N^r\tau-t')}\big(M_4^3(NG_N^{4;t'})+\tfrac1NS_{2}^3(NG_N^{2;t'})+\tfrac1NS_{3}^3(NG_N^{3;t'})\big)dt'.
\end{multline}
Appealing again to Lemmas~\ref{lem:prop-L} and~\ref{lem:prop-MS}, and to the a priori estimate of Lemma~\ref{lem:cum-ap} in the forms $\|NG_N^{4}\|_{\Ld^2_\beta(\D^2)}\lesssim_\beta N^{-\frac12}$, $\|NG_N^{2}\|_{\Ld^2_\beta(\D^2)}\lesssim_\beta N^\frac12$, and $\|NG_N^{3}\|_{\Ld^2_\beta(\D^2)}\lesssim_\beta 1$, the corresponding result follows.

\medskip\noindent
Finally, similar arguments show that
\begin{eqnarray*}
\bigg\|\int_0^{N^r\tau} e^{-iL_2(N^r\tau-t')}\big(S_1^2G_N^{1;t'}+M_3^2(NG_N^{3;t'})\big)\,dt'\bigg\|_{H^{-1}_\beta(\D^2)}
&\lesssim_\beta&\langle\tau\rangle^2N^{2r},\\
\bigg\|e^{-iN^rL_3\tau}(NG_N^{3;\circ})+\int_0^{N^r\tau} e^{-iL_3(N^r\tau-t')}S_1^3G_N^{1;t'}\,dt'\bigg\|_{H^{-1}_\beta(\D^3)}
&\lesssim_\beta&\langle\tau\rangle^2N^{2r},
\end{eqnarray*}
and the last estimate~\eqref{eq:opt-bound-order3} then follows from the previous descriptions of $G_N^2,G_N^3$.
\end{proof}

Solving the expressions for $NG_N^2$ and $NG_N^3$ in terms of $g_N^1$ and inserting them into the equation for $g_N^1$, we are led to a closed equation on $g_N^1$ as claimed in Proposition~\ref{prop:truncation}.

\begin{proof}[Proof of Proposition~\ref{prop:truncation}]
Since Lemmas~\ref{lem:prop-L} and~\ref{lem:prop-MS} yield for all $H\in\Ld^\infty(\R^+;C^\infty_c(\D^3))$,
\[\bigg\|\int_0^{N^r\tau}e^{-iL_2(N^r\tau-t')}M_3^2H^{t'}\,dt'\bigg\|_{H^{-2}_\beta(\D^2)}\,\lesssim_\beta\,\langle\tau\rangle^3N^{3r}\sup_{0\le t'\le N^r\tau}\|H^{t'}\|_{H^{-1}_\beta(\D^3)},\]
we may insert the approximate expression for $NG_N^3$ inside that for $NG_N^2$ in Lemma~\ref{lem:truncate-eqn-G123}, to the effect of
\begin{multline*}
\bigg\|NG_N^{2;N^r\tau}
-\int_0^{N^r\tau} e^{-iL_2(N^r\tau-t')}M_3^2e^{-iL_3t'}(NG_N^{3;\circ})\,dt'
-\int_0^{N^r\tau} e^{-iL_2(N^r\tau-t')}S_1^2G_N^{1;t'}\,dt'\\
-\int_0^{N^r\tau}\!\!\!\int_0^{t'} e^{-iL_2(N^r\tau-t')}M_3^2e^{-iL_3(t'-t'')}S_1^3G_N^{1;t''}\,dt''dt'
\bigg\|_{H^{-2}_\beta(\D^2)}
\,\lesssim_\beta\,\langle\tau\rangle^5N^{5r-\frac12}.
\end{multline*}
Next, we argue that we can replace $G_N^1$ by its initial condition $g^\circ$ (cf.~Lemma~\ref{lem:initial}) in the last two left-hand side terms.
For that purpose, we appeal to~\eqref{eq:G1} in the time-integrated form of
\begin{equation}\label{eq:G1-bis}
G_N^{1;N^r\tau}=g^\circ+\tfrac1N\int_0^{N^r\tau} e^{-iL_1(N^r\tau-t')}M_2^1(NG_N^{2;t'})\,dt',
\end{equation}
which yields, in view of Lemmas~\ref{lem:prop-L} and~\ref{lem:prop-MS} and of the a priori estimate of Lemma~\ref{lem:cum-ap} in the form $\|NG_N^{2}\|_{\Ld^2_\beta(\D^2)}\lesssim_\beta N^\frac12$,
\begin{equation}\label{eq:G1-bis+}
\|G_N^{1;N^r\tau}-g^\circ\|_{H^{-1}_\beta(\D)}\,\lesssim\,\langle\tau\rangle^2N^{2r-\frac12}.
\end{equation}
Since Lemmas~\ref{lem:prop-L} and~\ref{lem:prop-MS} yield for all $H\in\Ld^\infty(\R^+;C^\infty_c(\D))$,
\begin{equation*}
\bigg\|\int_0^{N^r\tau}e^{-iL_2(N^r\tau-t')}S_1^2H^{t'}\,dt'\bigg\|_{H^{-2}_\beta(\D^2)}\,\lesssim_\beta\,\langle\tau\rangle^3N^{3r}\sup_{0\le t'\le N^r\tau}\|H^{t'}\|_{H^{-1}_\beta(\D)},
\end{equation*}
and similarly,
\begin{multline*}
\bigg\|\int_0^{N^r\tau}\!\!\!\int_0^{t'}e^{-iL_2(N^r\tau-t')}M_3^2e^{-iL_3(t'-t'')}S_1^3H^{t''}\,dt''dt'\bigg\|_{H^{-3}_\beta(\D^2)}\\
\,\lesssim_\beta\,\langle\tau\rangle^7N^{7r}\sup_{0\le t'\le N^r\tau}\|H^{t'}\|_{H^{-1}_\beta(\D)},
\end{multline*}
we may insert~\eqref{eq:G1-bis+} into the above approximate expression for $NG_N^2$, to the effect of
\begin{multline*}
\bigg\|NG_N^{2;N^r\tau}
-\int_0^{N^r\tau} e^{-iL_2(N^r\tau-t')}M_3^2e^{-iL_3t'}(NG_N^{3;\circ})\,dt'
-\int_0^{N^r\tau} e^{-iL_2(N^r\tau-t')}S_1^2g^\circ\,dt'\\
-\int_0^{N^r\tau}\!\!\!\int_0^{t'} e^{-iL_2(N^r\tau-t')}M_3^2e^{-iL_3(t'-t'')}S_1^3g^\circ\,dt''dt'
\bigg\|_{H^{-3}_\beta(\D^2)}
\,\lesssim_\beta\,\langle\tau\rangle^9N^{9r-\frac12}.
\end{multline*}
Finally inserting this into the equation~\eqref{eq:G1} for $g_N^1$, and using Lemma~\ref{lem:prop-MS} once again, the conclusion follows.
\end{proof}

\section{Markovian limit of the truncated hierarchy}\label{sec:eqns-lim}

The present section is devoted to the computation of the long-time limit of the different terms in the integrated truncated BBGKY hierarchy of Proposition~\ref{prop:truncation}, showing that it coincides with the expected linearized Lenard--Balescu operator~\eqref{eq:FP}. The proof of Theorem~\ref{th:main} is concluded in Section~\ref{sec:proof-main}.

\subsection{Laplace transform}
The long-time propagators in Proposition~\ref{prop:truncation} are best computed formally by means of Laplace transforms. In such terms, we reformulate as follows the result of Proposition~\ref{prop:truncation}.

\begin{cor}\label{cor:truncation}
Let the assumptions of Theorem~\ref{th:main} hold for $V,\beta,f^\circ$.
Given $0\le r<\frac1{18}$, there holds for all $\phi\in C^\infty_c(\R^+)$,
\begin{align*}
&\bigg\|\int_0^\infty\phi(\tau)\,(N\partial_t g_N^{1})|_{t=N^r\tau}\,d\tau\\
&\hspace{0.5cm}-\int_\R g_\phi(\alpha)\bigg(\int_{\T^d}\Big(M_2^1\big(iL_2+\tfrac{i\alpha+1}{N^r}\big)^{-1}S_1^2g^\circ\Big)(x,\cdot)\,dx\\
&\hspace{1.3cm}+\int_{\T^d}\Big(M_2^1\big(iL_2+\tfrac{i\alpha+1}{N^r}\big)^{-1}M_3^2 \big(iL_3+\tfrac{i\alpha+1}{N^r}\big)^{-1}S_1^3g^\circ\Big)(x,\cdot)\,dx\\
&\hspace{1.3cm}+\tfrac{i\alpha+1}{N^r}\int_{\T^d}\Big(M_2^1\big(iL_2+\tfrac{i\alpha+1}{N^r}\big)^{-1}M_3^2\big(iL_3+\tfrac{i\alpha+1}{N^r}\big)^{-1}(NG_N^{3;\circ})\Big)(x,\cdot)\,dx
\bigg)\,d\alpha\bigg\|_{H^{-4}_\beta(\R^d)}\\
&\hspace{12.3cm}\lesssim_{\beta,\phi}\,N^{9r-\frac12},
\end{align*}
where $g_\phi(\alpha):=\frac1{2\pi}\int_0^\infty \frac{e^{(i\alpha+1)\tau}}{i\alpha+1}\phi(\tau)\,d\tau$ belongs to $C^\infty_b(\R)$ and satisfies
\[|g_\phi(\alpha)|\lesssim_\phi\langle\alpha\rangle^{-2}\qquad\text{and}\qquad\int_\R g_\phi=\int_0^\infty\phi.\qedhere\]
\end{cor}

Integrating the result of Proposition~\ref{prop:truncation} with a test function $\phi$ in time,
this corollary directly follows from applying the Laplace transform in form of the following product formula.

\begin{lem}[Product formula for Laplace transform]\label{lem:Laplace}
For all $\phi\in C^\infty_c(\R^+)$,
there holds for all $n\ge1$ and all generators $R_1,\ldots,R_n$ of uniformly bounded $C_0$-groups on a Hilbert space $\Hc$,
\begin{multline*}
\int_0^\infty \phi(\tau)\int_{(\R^+)^n}\mathds1_{t_1\le\ldots\le t_n\le N\tau}\,e^{-iR_n(N\tau-t_n)}\otimes\ldots\otimes e^{-iR_1(t_{2}-t_1)}dt_1\ldots dt_n\,d\tau\\
\,=\,\int_\R g_\phi(\alpha)\,\big(iR_n+\tfrac{i\alpha+1}{N}\big)^{-1}\otimes\ldots\otimes\big(iR_1+\tfrac{i\alpha+1}{N}\big)^{-1}d\alpha,
\end{multline*}
where the transformation $g_\phi$ is as in the statement of Corollary~\ref{cor:truncation} above.
\end{lem}

\begin{proof}
Consider the increments $\tau_j:=t_{j+1}-t_j$ for $1\le j\le n-1$, and set $\tau_n:=N\tau-t_n$ and $\tau_0:=t_1$.
Using the uniform boundedness of the $C_0$-groups, we can write
\begin{multline*}
\int_0^\infty \phi(\tau)\int_{(\R^+)^n}\mathds1_{t_1\le\ldots\le t_n\le N\tau}\,e^{-iR_n(N\tau-t_n)}\otimes\ldots\otimes e^{-i R_1(t_{2}-t_1)}dt_1\ldots dt_n\,d\tau\\
\,=\,\int_0^\infty \phi(\tau)\int_{(\R^+)^{n+1}}\delta(\tau_0+\ldots+\tau_n-N\tau)\,e^{-iR_n\tau_n}\otimes\ldots\otimes e^{-iR_1\tau_1}d\tau_0\ldots d\tau_n\,d\tau,
\end{multline*}
or equivalently,
\begin{multline*}
\int_0^\infty \phi(\tau)\int_{(\R^+)^n}\mathds1_{t_1\le\ldots\le t_n\le N\tau}\,e^{-iR_n(N\tau-t_n)}\otimes\ldots\otimes e^{-i R_1(t_{2}-t_1)}dt_1\ldots dt_n\,d\tau\\
\,=\,\int_0^\infty e^\tau\phi(\tau)\int_{(\R^+)^{n+1}}\delta(\tau_0+\ldots+\tau_n-N\tau)\\
\times\,e^{-(iR_n+\frac1N)\tau_n}\otimes\ldots\otimes e^{-(iR_1+\frac1N)\tau_1}\,e^{-\frac1N\tau_0}\,d\tau_0\ldots d\tau_n\,d\tau.
\end{multline*}
Then inserting the formula $\delta(\tau_0+\ldots+\tau_n-N\tau)=\frac1{2\pi}\int_\R e^{-i\alpha(\tau_0+\ldots+\tau_n-N\tau)}d\alpha$, we find
\begin{eqnarray*}
\lefteqn{\int_0^\infty \phi(\tau)\int_{(\R^+)^n}\mathds1_{t_1\le\ldots\le t_n\le N\tau}\,e^{-iR_n(N\tau-t_n)}\otimes\ldots\otimes e^{-i R_1(t_{2}-t_1)}dt_1\ldots dt_n\,d\tau}\\
&=&\frac1{2\pi}\int_\R\Big(\int_0^\infty e^{(i\alpha+1)\tau}\phi(\tau)\,d\tau\Big)\Big(\tfrac1N\int_0^\infty e^{-\frac1N(i\alpha+1)\tau_0}d\tau_0\Big)\\
&&\hspace{2cm}\times\Big(\int_0^\infty e^{-(iR_n+\frac{i\alpha+1}N)\tau_n}d\tau_n\Big)\otimes\ldots\otimes\Big(\int_0^\infty e^{-(iR_1+\frac{i\alpha+1}N)\tau_1}d\tau_1\Big)\,d\alpha\\
&=&\frac1{2\pi}\int_\R\Big(\int_0^\infty \frac{e^{(i\alpha+1)\tau}}{i\alpha+1}\phi(\tau)\,d\tau\Big)\big(iR_n+\tfrac{i\alpha+1}{N}\big)^{-1}\otimes\ldots\otimes\big(iR_1+\tfrac{i\alpha+1}{N}\big)^{-1}d\alpha.
\end{eqnarray*}
It remains to analyze the integral in bracket, that is, $g_\phi(\alpha):=\frac1{2\pi}\int_0^\infty \frac{e^{(i\alpha+1)\tau}}{i\alpha+1}\phi(\tau)\,d\tau$.
Note that
\begin{eqnarray*}
(1+\alpha^2)\int_0^\infty \frac{e^{(i\alpha+1)\tau}}{i\alpha+1}\phi(\tau)\,d\tau&=&2\int_0^\infty e^{(i\alpha+1)\tau}\phi(\tau)d\tau+\int_0^\infty (-\partial_\tau) e^{(i\alpha+1)\tau}\phi(\tau)d\tau\\
&=&2\int_0^\infty e^{(i\alpha+1)\tau}\phi(\tau)d\tau+\int_0^\infty e^{(i\alpha+1)\tau}\phi'(\tau)d\tau+\phi(0),
\end{eqnarray*}
where the right-hand side is uniformly bounded in $\alpha$ for $\phi\in C^\infty_c(\R^+)$, hence $g_\phi(\alpha)$ is bounded by $C_\phi(1+\alpha^2)^{-1}$.
Moreover, a straightforward computation by means of Fourier transforms yields the identity $\int_\R g_\phi=\int_0^\infty\phi$.
\end{proof}

\subsection{Preliminary estimates}
We establish the following uniform estimates, which are useful for application of Lebesgue's dominated convergence theorem when computing the limit of the different terms appearing in Corollary~\ref{cor:truncation}. In particular, the bound on the dispersion function $\e_{\beta,m}^\circ$ improves on the positivity statement of Lemma~\ref{lem:spectrum}(ii).

\begin{lem}\label{lem:prel-est}
The following hold for all $\beta\in(0,\infty)$ and $k\in2\pi\Z^d\setminus\{0\}$,
\begin{enumerate}[(i)]
\item \emph{Decay estimate:} for all $\Im\omega>0$,
\[\big|\big\langle\tfrac1{k\cdot v-\omega}\big\rangle_{v}\big|+\tfrac1{|k|}\big|\big\langle\tfrac{k\cdot v}{k\cdot v-\omega}\big\rangle_{v}\big|\,\lesssim_\beta\,\tfrac1{1+|\Re\omega|};\]
\item \emph{Lower bound on $\e_{\beta,m}^\circ$:} given $V\in\Ld^\infty(\T^d)$ with $\widehat V\ge0$ and $\beta\|V\|_{\Ld^\infty}\le \frac1{C_0}$ for some large enough constant $C_0\simeq1$, there holds for $1\le m\le N$ and $\Im\omega>0$,
\[|\e_{\beta,m}^\circ(k,\omega)|\,\gtrsim\,1;\]
\item \emph{Uniform bounds:} for all $\Im\omega,\Im\eta>0$,
\begin{gather*}
\big\langle\big|\tfrac{1}{k\cdot v-\omega}\big|\big\rangle_v+\tfrac1{|k|}\big\langle\big|\tfrac{k\cdot v}{k\cdot v-\omega}\big|\big\rangle_v\,\lesssim_\beta\,1+\log(1+\tfrac1{\Im\omega}),\\
\big|\big\langle\tfrac{1}{(k\cdot v-\eta)(k\cdot v-\omega)}\big\rangle_v\big|+\tfrac1{|k|}\big|\big\langle\tfrac{k\cdot v}{(k\cdot v-\eta)(k\cdot v-\omega)}\big\rangle_v\big|\,\lesssim_\beta\,1.\qedhere
\end{gather*}
\end{enumerate}
\end{lem}

\begin{proof}
We start with the proof of~(i). Setting $\hat k:={k}/{|k|}$ and splitting the $v$-integral over $\hat k\R$ and over $(\hat k\R)^\bot$, we can write
\[\big\langle\tfrac1{k\cdot v-\omega}\big\rangle_{v}\,=\,\big(\tfrac{\beta}{2\pi|k|^2}\big)^\frac12\int_{\R}\tfrac1{y-\omega}\,e^{-\frac\beta{2|k|^2} y^2}\,dy.\]
Decomposing the real and imaginary parts, an elementary computation yields for all $r,c>0$ and $y_0\in\R$,
\begin{eqnarray}\label{eq:comput-gen}
\Big|c^\frac12\int_\R\tfrac1{y-y_0-ir}e^{-cy^2}dy\Big|&\le&\Big|c^\frac12\int_\R\tfrac{y-y_0}{(y-y_0)^2+r^2}e^{-cy^2}dy\Big|+c^\frac12\int_\R\tfrac{r}{(y-y_0)^2+r^2}e^{-cy^2}dy\nonumber\\
&\lesssim&\big(\tfrac1{|y_0|}+c^\frac12 e^{-\frac c2y_0^2}\big)\wedge\big(1+c^\frac12\big)\nonumber\\
&\lesssim&\tfrac1{|y_0|}\wedge\big(1+c^\frac12\big),
\end{eqnarray}
and we deduce
\[\big|\big\langle\tfrac1{k\cdot v-\omega}\big\rangle_{v}\big|\,\lesssim\,\tfrac1{|\Re\omega|}\wedge\big(1+\tfrac\beta{|k|^2}\big)^\frac12\,\lesssim_\beta\,\tfrac1{1+|\Re\omega|}.\]
Similarly noting that
\[c^\frac12\Big|c^\frac12\int_\R\tfrac{1}{y-y_0-ir}ye^{-cy^2}dy\Big|\,\lesssim\,\tfrac1{|y_0|}\wedge\big(1+c^\frac12\big),\]
the corresponding bound on $\langle\tfrac{k\cdot v}{k\cdot v-\omega}\rangle_{v}$ follows.

\medskip\noindent
We turn to the lower bound~(ii) for $\e_{\beta,m}^\circ$.
Since $\widehat V$ is nonnegative, we deduce from~\eqref{eq:bound-eps-pre} that
\[|\e_{\beta,m}^\circ(k,\omega)|\,\ge\,1-\tfrac{N+1-m}N\widehat V(k)|\Re\omega|\big|\big\langle\tfrac{\beta (k\cdot v-\Re\omega)}{(k\cdot v-\Re\omega)^2+(\Im\omega)^2}\big\rangle_{v}\big|,\]
hence, in view of~\eqref{eq:comput-gen},
\[|\e_{\beta,m}^\circ(k,\omega)|\,\ge\,1-C\beta\widehat V(k)
\,\ge\,1-C(2\pi)^d\beta\|V\|_{\Ld^\infty(\T^d)},\]
and the claim follows.

\medskip\noindent
It remains to establish the bounds in~(iii) and we start with the first one. Writing
\[\big\langle\big|\tfrac{1}{k\cdot v-\omega}\big|\big\rangle_v\,\lesssim\,\big(\tfrac\beta{2\pi|k|^2}\big)^\frac12\int_\R\tfrac1{|y-\Re\omega|+|\Im\omega|}\,e^{-\frac\beta{2|k|^2}y^2}\,dy,\]
and separately estimating the contribution of the $y$-integral for $|y-\Re\omega|\le\Im\omega$, for $\Im\omega\le|y-\Re\omega|\le L$, and for $|y-\Re\omega|\ge L$, we deduce for all $L\ge\Im\omega$,
\[\big\langle\big|\tfrac{1}{k\cdot v-\omega}\big|\big\rangle_v\,\lesssim\,\beta^\frac12+\tfrac1L+\beta^\frac12\int_{\Im\omega\le|y-\Re\omega|\le L}\tfrac1{|y-\Re\omega|}\,dy\,\lesssim\,\beta^\frac12+\tfrac1L+\beta^\frac12\log\tfrac{L}{\Im\omega}.\]
Choosing $L=1+\Im\omega$, this yields the claim
\[\big\langle\big|\tfrac{1}{k\cdot v-\omega}\big|\big\rangle_v\,\lesssim_\beta\,1+\log(1+\tfrac1{\Im\omega}).\]
Multiplying the integrand by $|k\cdot v|$ only changes the estimate by a factor $|k|$, and the first part of~(iii) follows.

\medskip\noindent
We turn to the second part of~(iii). Writing
\[\big\langle\tfrac{1}{(k\cdot v-\eta)(k\cdot v-\omega)}\big\rangle_v\,=\,\big(\tfrac\beta{2\pi|k|^2}\big)^\frac12\int_\R\tfrac1{(y-\eta)(y-\omega)}\,e^{-\frac\beta{2|k|^2}y^2}\,dy,\]
and slightly deforming the integration path for $y$ close to $\Re\eta$ and to $\Re\omega$ in order to ensure that $|y-\eta|$ and $|y-\omega|$ are uniformly bounded below by $1$ (note that $\eta$ and $\omega$ are on the same complex half-plane), the result follows.
\end{proof}

\subsection{Contribution from $2$-particle correlations}
This section is devoted to the explicit computation of the contribution of $2$-particle correlations in the formula of Corollary~\ref{cor:truncation}, which formally takes the form $M_2^1(iL_2+0)^{-1}S_1^2g^\circ$.

\begin{prop}\label{prop:2-corr-comp}
Given $V\in W^{1,\infty}(\T^d)$ and $\beta\in(0,\infty)$, the following convergence holds in $H^{-1}_\beta(\R^d)$, uniformly for $N\ge1$,
\begin{multline*}
\lim_{\omega\to0\atop\Im\omega>0}\int_{\T^d}\big(M_2^1(iL_2-i\omega)^{-1}S_1^2g^\circ\big)(x,v)\,dx\\
\,=\,(\nabla_{v}-\beta v)\cdot\bigg(\sum_{k\in2\pi\Z^d}(k\otimes k)\tfrac{\pi\widehat V(k)^2\langle\delta(k\cdot(v_*-v))\rangle_{v_*}}{|\e_{\beta,1}^\circ(k,k\cdot v+i0)|^2}\bigg)\nabla_v g^\circ(v)\\
+(\nabla_{v}-\beta v)\cdot\bigg(\sum_{k\in2\pi\Z^d}\beta\widehat V(k)\,(k\otimes k)\tfrac{\pi\widehat V(k)^2\langle\delta(k\cdot(v_*-v))\rangle_{v_*}}{|\e_{\beta,1}^\circ(k,k\cdot v+i0)|^2}\bigg)(\nabla_v-\beta v) g^\circ(v),
\end{multline*}
where in addition the argument of the limit can be written as $(\nabla_v-\beta v)\cdot G_{N,\beta}^\omega(v)$ where $G_{N,\beta}^\omega(v)$ is bounded pointwise by $C_\beta|\langle\nabla_v-\beta v\rangle g^\circ(v)|$ uniformly for~$\Im\omega>0$.
\end{prop}

We first need to find a way to explicitly compute the resolvent of $L_2=L_2^{(0)}+L_2^{(1)}$. As the resolvents of both summands are explicitly given in Lemma~\ref{lem:spectrum}(ii), the resolvent of their sum can be deduced from the following useful general identity (see e.g.~\cite[p.120]{Reed-Simon-73}).

\begin{lem}[Resolvent of sums of commuting operators]\label{lem:sum-resolv}
Let $iH_1$ and $iH_2$ denote two generators of uniformly bounded commuting $C_0$-groups on a Hilbert space $\Hc$.
Then, for all $0<\Im\eta<\Im\omega$,
\[(H_1+H_2-\omega)^{-1}=\frac1{2\pi i}\int_\R(H_1+\alpha-\omega+\eta)^{-1}(H_2-\alpha-\eta)^{-1}\,d\alpha.\qedhere\]
\end{lem}

\begin{proof}
As the two generators commute, their sum also generate a $C_0$-group, and its resolvent is given as the Laplace transform of the generated group,
\[(iH_1+iH_2-i\omega)^{-1}=\int_0^\infty e^{-(iH_1+iH_2)t}e^{i\omega t}\,dt=\int_0^\infty e^{-iH_1t}e^{-iH_2t}e^{i\omega t}\,dt,\]
hence, for all $\eta$,
\[(iH_1+iH_2-i\omega)^{-1}=\int_0^\infty e^{-(iH_1-i\omega+i\eta)t}e^{-(iH_2-i\eta)t}\,dt.\]
Inserting the formula $\delta(t-t')=\frac1{2\pi}\int_\R e^{-i\alpha(t-t')}d\alpha$ and invoking the uniform boundedness of the $C_0$-groups, we can write for $0<\Im\eta<\Im\omega$,
\[(iH_1+iH_2-i\omega)^{-1}=\frac1{2\pi}\int_\R\Big(\int_0^\infty e^{-(iH_1+i\alpha-i\omega+i\eta)t}\,dt\Big)\Big(\int_0^\infty e^{-(iH_2-i\alpha-i\eta)t'}\,dt'\Big)\,d\alpha,\]
and the conclusion follows.
\end{proof}

With this useful trick at hand, we may explicitly compute the resolvent of $L_2$ as required for the proof of Proposition~\ref{prop:2-corr-comp}.

\begin{proof}[Proof of Proposition~\ref{prop:2-corr-comp}]
By definition of $M_2^1$ in Lemma~\ref{lem:eqns}, we can write in Fourier variables,
\begin{multline*}
\int_{\T^d}\big(M_2^1(iL_2-i\omega)^{-1}S_1^2g^\circ\big)(x,v)\,dx\,=\,\big(\widehat M_2^1(i\widehat L_2-i\omega)^{-1}\widehat S_1^2 g^\circ\big)(0,v)\\
\,=\,\sum_{k\in2\pi\Z^d}ik\widehat V(k)\cdot(\nabla_{v}-\beta v)\Big\langle\big((i\widehat L_2-i\omega)^{-1}\widehat S_1^2g^\circ\big)((-k,v),(k,v_*))\Big\rangle_{v_*},
\end{multline*}
hence, by symmetry,
\begin{multline}\label{eq:comput-M21-g2}
\int_{\T^d}\big(M_2^1(iL_2-i\omega)^{-1}S_1^2g^\circ\big)(x,v)\,dx\\
\,=\,-\tfrac1{2i}\sum_{k\in2\pi\Z^d}k\widehat V(k)\cdot(\nabla_{v}-\beta v)
\Big\langle\big((i\widehat L_2-i\omega)^{-1}\widehat S_1^2g^\circ\big)((-k,v),(k,v_*))\\
-\big((i\widehat L_2-i\omega)^{-1}\widehat S_1^2g^\circ\big)((k,v),(-k,v_*))\Big\rangle_{v_*}.
\end{multline}
In view of Lemma~\ref{lem:sum-resolv} with $L_2=L_2^{(0)}+L_2^{(1)}$, we can write
\[(i\widehat L_2-i\omega)^{-1}\widehat S_1^2g^\circ=-\frac1{2\pi}\int_\R\big(\widehat L_2^{(1)}+\alpha-\tfrac\omega2\big)^{-1}\big(\widehat L_2^{(0)}-\alpha-\tfrac\omega2\big)^{-1}\widehat S_1^2g^\circ\,d\alpha.\]
Inserting formulas for resolvents as given in Lemma~\ref{lem:spectrum}(ii), we deduce
\begin{multline}\label{eq:pre-compl-def}
\Big\langle \big((i\widehat L_2-i\omega)^{-1}\widehat S_1^2g^\circ\big)((-k,v),(k,v_*))\Big\rangle_{v_*}\\
=\frac1{2\pi}\int_\R\tfrac{1}{\e_{\beta,1}^\circ(k,\frac\omega2-\alpha)}\tfrac1{k\cdot v+\alpha+\frac\omega2}\Big\langle\tfrac{\widehat S_1^2g^\circ((-k,v),(k,v_*))}{k\cdot v_*+\alpha-\frac\omega2}\Big\rangle_{v_*}\,d\alpha.
\end{multline}
As $\Im\omega>0$, we note that the integrand
\[\alpha~\mapsto~ \tfrac{1}{\e_{\beta,1}^\circ(k,\frac\omega2-\alpha)}\Big\langle\tfrac{\widehat S_1^2g^\circ((-k,v),(k,v_*))}{k\cdot v_*+\alpha-\frac\omega2}\Big\rangle_{v_*}\]
is analytic on the lower complex half-plane $\Im\alpha<\frac12\Im\omega$.
In addition, in view of Lemma~\ref{lem:prel-est} and in view of the definition of $S_1^2$ in Lemma~\ref{lem:eqns} in the form
\[\widehat S_1^{2}g^\circ((-k,v),(k,v_*))=-ik\widehat V(k)\cdot(\nabla_{v}-\beta v+\beta v_*)\,g^\circ(v),\]
the integrand is bounded by
\[\Big|\tfrac{1}{\e_{\beta,1}^\circ(k,\frac\omega2-\alpha)}\Big\langle\tfrac{\widehat S_1^2g^\circ((-k,v),(k,v_*))}{k\cdot v_*+\alpha-\frac\omega2}\Big\rangle_{v_*}\Big|\,\lesssim_\beta\,\tfrac1{1+|\Re(\alpha-\frac\omega2)|}\big(|(\nabla_v-\beta v)g^\circ(v)|+|g^\circ(v)|\big).\]
Complex deformation can then be applied to~\eqref{eq:pre-compl-def} in the lower complex half-plane and we are led to the residue at $\alpha=-k\cdot v-\frac\omega2$,
\begin{equation}\label{eq:L2-inv}
\Big\langle \big((i\widehat L_2-i\omega)^{-1}\widehat S_1^2g^\circ\big)((-k,v),(k,v_*))\Big\rangle_{v_*}
=\tfrac{-i}{\e^\circ_{\beta,1}(k,k\cdot v+\omega)}\Big\langle\tfrac{\widehat S_1^2g^\circ((-k,v),(k,v_*))}{k\cdot (v_*- v)-\omega}\Big\rangle_{v_*}.
\end{equation}
Using the definition of $\e^\circ_{\beta,1}$ in Lemma~\ref{lem:spectrum}(ii) in the form
\begin{equation}\label{eq:decomp-1/eps}
\tfrac1{\e^\circ_{\beta,1}(k,k\cdot v+\omega)}=\tfrac1{|\e^\circ_{\beta,1}(k,k\cdot v+\omega)|^2}\big(1+\beta\widehat V(k)\big\langle\tfrac{ k\cdot v_*}{k\cdot (v_*-v)-\bar\omega}\big\rangle_{v_*}\big),
\end{equation}
inserting the above definition of $S_1^2$, and 
reorganizing the term, we find
\begin{multline*}
\Big\langle \big((i\widehat L_2-i\omega)^{-1}\widehat S_1^2g^\circ\big)((-k,v),(k,v_*))\Big\rangle_{v_*}\\
=-\tfrac{\widehat V(k)}{|\e^\circ_{\beta,1}(k,k\cdot v+\omega)|^2}\Big(1+\beta\widehat V(k)\big\langle\tfrac{ k\cdot v_*}{k\cdot(v_*-v)-\bar\omega}\big\rangle_{v_*}\Big)\\
\times\Big(\big\langle\tfrac{1}{k\cdot (v_*- v)-\omega}\big\rangle_{v_*}k\cdot\nabla g^\circ(v)+\big\langle\tfrac{k\cdot (v_*-v)}{k\cdot (v_*- v)-\omega}\big\rangle_{v_*}\beta g^\circ(v)\Big).
\end{multline*}
We note that Lemma~\ref{lem:prel-est} yields the following bound, uniformly for $\Im\omega>0$,
\[\Big|\Big\langle \big((i\widehat L_2-i\omega)^{-1}\widehat S_1^2g^\circ\big)((-k,v),(k,v_*))\Big\rangle_{v_*}\Big|\,
\lesssim_\beta\,\widehat V(k)\big(|(\nabla-\beta v)g^\circ(v)|+|g^\circ(v)|\big),\]
which implies the stated estimate.
Letting $\omega\to0$ with $\Im\omega>0$, we find
\begin{multline*}
\lim_{\omega\to0\atop\Im\omega>0}\Big\langle \big((i\widehat L_2-i\omega)^{-1}\widehat S_1^2g^\circ\big)((-k,v),(k,v_*))\Big\rangle_{v_*}\\
=-\tfrac{\widehat V(k)}{|\e^\circ_{\beta,1}(k,k\cdot v+i0)|^2}\Big(1+\beta\widehat V(k)\big\langle\tfrac{ k\cdot v_*}{k\cdot(v_*-v)+i0}\big\rangle_{v_*}\Big)
\Big(\big\langle\tfrac{1}{k\cdot (v_*- v)-i0}\big\rangle_{v_*}k\cdot\nabla g^\circ(v)+\beta g^\circ(v)\Big),
\end{multline*}
where all terms indeed make sense.
Noting that $\e^\circ_{\beta,1}(-k,-k\cdot v+i0)=\overline{\e^\circ_{\beta,1}(k,k\cdot v+i0)}$, we can write by symmetry
\begin{multline*}
\lim_{\omega\to0\atop\Im\omega>0}\tfrac1{2i}\Big\langle \big((i\widehat L_2-i\omega)^{-1}\widehat S_1^2g^\circ\big)((-k,v),(k,v_*))-\big((i\widehat L_2-i\omega)^{-1}\widehat S_1^2g^\circ\big)((k,v),(-k,v_*))\Big\rangle_{v_*}\\
=-\tfrac{\widehat V(k)}{|\e^\circ_{\beta,1}(k,k\cdot v+i0)|^2}\Im\Big[\Big(1+\beta\widehat V(k)\big\langle\tfrac{ k\cdot v_*}{k\cdot(v_*-v)+i0}\big\rangle_{v_*}\Big)
\Big(\big\langle\tfrac{1}{k\cdot (v_*- v)-i0}\big\rangle_{v_*}k\cdot\nabla g^\circ(v)+\beta g^\circ(v)\Big)\Big].
\end{multline*}
Decomposing
\[\big\langle\tfrac{\beta k\cdot v_*}{k\cdot(v_*-v)+i0}\big\rangle_{v_*}=\beta+\beta k\cdot v\,\big\langle\tfrac{1}{k\cdot(v_*-v)+i0}\big\rangle_{v_*},\]
and using the Sokhotski-Plemelj formula in the form
\[\Im\big\langle\tfrac{1}{k\cdot(v_*-v)-i0}\big\rangle_{v_*}=\pi\big\langle\delta\big(k\cdot(v_*-v)\big)\big\rangle_{v_*},\]
we find after straightforward simplifications,
\begin{multline*}
\lim_{\omega\to0\atop\Im\omega>0}\frac1{2i}\Big\langle \big((i\widehat L_2-i\omega)^{-1}\widehat S_1^2g^\circ\big)((-k,v),(k,v_*))-\big((i\widehat L_2-i\omega)^{-1}\widehat S_1^2g^\circ\big)((k,v),(-k,v_*))\Big\rangle_{v_*}\\
=-\tfrac{\pi\widehat V(k)\langle\delta(k\cdot(v_*-v))\rangle_{v_*}}{|\e^\circ_{\beta,1}(k,k\cdot v+i0)|^2}\Big[k\cdot\nabla g^\circ(v)+\beta\widehat V(k)k\cdot(\nabla_v-\beta v) g^\circ(v)\Big].
\end{multline*}
Inserting this into~\eqref{eq:comput-M21-g2} yields the conclusion.
\end{proof}

\subsection{Contribution from $3$-particle correlations}
We turn to the explicit computation of the contribution of $3$-particle correlations in the formula of Corollary~\ref{cor:truncation}, which formally takes the form $M_2^1(iL_2+0)^{-1}M_3^2(iL_3+0)^{-1}S_1^3g^\circ$.

\begin{prop}\label{prop:3-corr-comp}
Given $V\in W^{1,\infty}(\T^d)$ and $\beta\in(0,\infty)$, the following convergence holds in $H^{-1}_\beta(\R^d)$, uniformly for $N\ge1$,
\begin{multline*}
\lim_{\omega\to0\atop\Im\omega>0}\int_{\T^d}\big(M_2^1(iL_2-i\omega)^{-1}M_3^2 (iL_3-i\omega)^{-1}S_1^3g^\circ\big)(x,v)\,dx\\
\,=\,-(\nabla_{v}-\beta v)\cdot\bigg(\sum_{k\in2\pi\Z^d}\beta\widehat V(k)(k\otimes k)\tfrac{\pi\widehat V(k)^2\langle\delta(k\cdot (v_*- v))\rangle_{v_*}}{|\e^\circ_{\beta,1}(k,k\cdot v+i0)|^2}\Big(\tfrac{N-1}N\tfrac{1+\beta \widehat V(k)}{1+\frac{N-1}N\beta\widehat V(k)}\Big)\bigg)\\
\times(\nabla_v-\beta v)g^\circ(v),
\end{multline*}
where in addition the argument of the limit can be written as $(\nabla_v-\beta v)G_{N,\beta}^\omega(v)$ where $G_{N,\beta}^\omega(v)$ is bounded pointwise by $C_\beta|(\nabla_v-\beta v)g^\circ(v)|$ uniformly for $\Im\omega>0$.
\end{prop}

\begin{proof}
By definition of $M_2^1$ in Lemma~\ref{lem:eqns}, we can write in Fourier variables,
\begin{eqnarray}
\lefteqn{\int_{\T^d}\big(M_2^1(iL_2-i\omega)^{-1}M_3^2 (iL_3-i\omega)^{-1}S_1^3g^\circ\big)(x,v)\,dx}\label{eq:g3-dec1}\\
&=&\big(\widehat M_2^1(i\widehat L_2-i\omega)^{-1}\widehat M_3^2 (i\widehat L_3-i\omega)^{-1}\widehat S_1^3g^\circ\big)(0,v)\nonumber\\
&=&\sum_{k\in2\pi\Z^d}ik\widehat V(k)
\cdot(\nabla_{v}-\beta v)\Big\langle\big((i\widehat L_2-i\omega)^{-1}\widehat M_3^2 (i\widehat L_3-i\omega)^{-1}\widehat S_1^3g^\circ\big)((-k,v),(k,v_*))\Big\rangle_{v_*},\nonumber
\end{eqnarray}
while the definition of $M_3^2$ yields
\begin{align}
&\big(\widehat M_3^2 (i\widehat L_3-i\omega)^{-1}\widehat S_1^3g^\circ\big)((-k,v),(k,v_*))\label{eq:g3-dec2}\\
&\hspace{-0.2cm}\,=\,\tfrac{N-1}N\sum_{k'\in2\pi\Z^d}ik'\widehat V(k')
\cdot(\nabla_{v}-\beta v)
\Big(\big\langle\big((i\widehat L_3-i\omega)^{-1}\widehat S_1^3g^\circ\big)\big((-k-k',v),(k,v_*),(k',w_*)\big)\big\rangle_{w_*}\nonumber\\
&\hspace{5cm}+\big\langle\big((i\widehat L_3-i\omega)^{-1}\widehat S_1^3g^\circ\big)\big((-k,v),(k-k',v_*),(k',w_*)\big)\big\rangle_{w_*}\Big).\nonumber
\end{align}
In view of the prefactors $k,k'$, we can restrict the sums to $k,k'\ne0$.
We start with the evaluation of the resolvent $(i\widehat L_3-i\omega)^{-1}$. For that purpose, we note that the definitions of $S_1^3$ and $L_3$ in Lemma~\ref{lem:eqns} yield for $k,k'\ne0$,
\begin{eqnarray*}
(\widehat S_1^3g^\circ)\big((-k-k',v),(k,v_*),(k',w_*)\big)&=&-\delta(k+k')\,ik\cdot (v_*-w_*)\beta\widehat V(k)g^\circ(v),\\
(\widehat S_1^3g^\circ)\big((-k,v),(k-k',v_*),(k',w_*)\big)&=&\delta(k)\,ik'\cdot (v_*-w_*)\beta\widehat V(k')g^\circ(v)~=~0,
\end{eqnarray*}
and for $\widehat G=\widehat G((0,v),(k,v_*),(-k,w_*))$,
\[\widehat L_{3}\widehat G=k\cdot v_*\big(\widehat G+\tfrac{N-1}N\beta\widehat V(k)\langle \widehat G\rangle_{v_*}\big)-k\cdot w_*\big(\widehat G+\tfrac{N-1}N\beta\widehat V(k)\langle \widehat G\rangle_{w_*}\big).\]
Hence, for
\[\widehat H^3\big((k_0,v_0),(k_1,v_1),(k_2,v_2)\big):=\tfrac{\frac\beta2(\widehat V(k_1)+\widehat V(k_2))}{1+\frac{N-1}N\frac\beta2(\widehat V(k_1)+\widehat V(k_2))}g^\circ(v_0),\]
comparing the above formulas for $S_1^3$ and $L_3$ yields
\begin{multline*}
(\widehat L_3\widehat H^3)\big((0,v),(k,v_*),(-k,w_*)\big)=k\cdot (v_*-w_*)\beta\widehat V(k)g^\circ(v)\\
=i(\widehat S_1^3g^\circ)\big((0,v),(k,v_*),(-k,w_*)\big).
\end{multline*}
Inserting these computations into~\eqref{eq:g3-dec2}, we obtain for $k\ne0$,
\begin{multline}\label{eq:g3-dec3}
\big(\widehat M_3^2 (i\widehat L_3-i\omega)^{-1}\widehat S_1^3g^\circ\big)((-k,v),(k,v_*))\\
\,=\,\tfrac{N-1}Nik\widehat V(k)
\cdot(\nabla_{v}-\beta v) \big\langle\big((i\widehat L_3-i\omega)^{-1}i\widehat L_3\widehat H^3\big)\big((0,v),(k,v_*),(-k,w_*)\big)\big\rangle_{w_*}.
\end{multline}
We now turn back to~\eqref{eq:g3-dec1}: repeating the computation of the resolvent of $L_2$ as in~\eqref{eq:L2-inv}, we find
\begin{multline*}
\int_{\T^d}\big(M_2^1(iL_2-i\omega)^{-1}M_3^2 (iL_3-i\omega)^{-1}S_1^3g^\circ\big)(x,v)\,dx\\
\,=\,(\nabla_{v}-\beta v)\cdot\sum_{k\in2\pi\Z^d}\tfrac{k\widehat V(k)}{\e^\circ_{\beta,1}(k,k\cdot v+\omega)}\Big\langle\tfrac{(\widehat M_3^2 (i\widehat L_3-i\omega)^{-1}\widehat S_1^3g^\circ)((-k,v),(k,v_*))}{k\cdot (v_*- v)-\omega}\Big\rangle_{v_*},
\end{multline*}
hence, inserting~\eqref{eq:g3-dec3},
\begin{multline}\label{eq:g3-dec4}
\int_{\T^d}\big(M_2^1(iL_2-i\omega)^{-1}M_3^2 (iL_3-i\omega)^{-1}S_1^3g^\circ\big)(x,v)\,dx\\
\,=\,(\nabla_{v}-\beta v)\cdot\sum_{k\in2\pi\Z^d}\tfrac{N-1}N\tfrac{i(k\otimes k)\widehat V(k)^2}{\e^\circ_{\beta,1}(k,k\cdot v+\omega)}\Big\langle\tfrac1{k\cdot (v_*- v)-\omega}(\nabla_{v}-\beta v)\\
\times\big((i\widehat L_3-i\omega)^{-1}i\widehat L_3\widehat H^3\big)\big((0,v),(k,v_*),(-k,w_*)\big)\Big\rangle_{v_*,w_*}.
\end{multline}
Decomposing
\begin{equation}\label{eq:decomp-L3}
(i\widehat L_3-i\omega)^{-1}i\widehat L_3\widehat H^3\,=\,\widehat H^3+i\omega(i\widehat L_3-i\omega)^{-1}\widehat H^3,
\end{equation}
spectral calculus ensures that $(i\widehat L_3-i\omega)^{-1}i\widehat L_3\widehat H^3$ converges to $\widehat H^3-\widehat H^3_\circ$ in $\Ld^2_\beta(\D^3)$ as $\omega\to0$ with $\Im\omega>0$, where
\[\widehat H^3_\circ:=\delta(k_0)\delta(k_1)\delta(k_2)\widehat H^3\]
is the orthogonal projection of $\widehat H^3$ onto $\operatorname{Ker}L_3=\{\psi\in\Ld^2_\beta(\D^3):\nabla_{x_j}\psi\equiv0~\forall j\}$, and we note that the contribution of $\widehat H^3_\circ$ in~\eqref{eq:g3-dec4} vanishes. The singularity of the prefactor $\tfrac1{k\cdot (v_*- v)-\omega}$ in~\eqref{eq:g3-dec4} however forces us to proceed to a more careful analysis.
Explicitly computing the resolvent of $L_3$ based on Lemmas~\ref{lem:spectrum}(ii) and~\ref{lem:sum-resolv}, and inserting the definition of $H^3$, we find
\begin{multline*}
\Big\langle\tfrac1{k\cdot (v_*- v)-\omega}(\nabla_{v}-\beta v)\big((i\widehat L_3-i\omega)^{-1}\widehat H^3\big)\big((0,v),(k,v_*),(-k,w_*)\big)\Big\rangle_{v_*,w_*}\\
=\tfrac{\beta\widehat V(k)}{1+\frac{N-1}N\beta\widehat V(k)}(\nabla_v-\beta v)g^\circ(v)\bigg(-i\Big\langle\tfrac1{\e_{\beta,2}^\circ(-k,\omega-k\cdot v_*)}\tfrac1{k\cdot (v_*- v)-\omega}\big\langle\tfrac1{k\cdot(v_*-w_*-\omega)}\big\rangle_{w_*}\Big\rangle_{v_*}\\
+\tfrac{N-1}N\beta\widehat V(k)\tfrac1{2\pi}\int_\R\tfrac1{\e_{\beta,2}^\circ(-k,\frac\omega2-\alpha)\e_{\beta,2}^\circ(k,\frac\omega2+\alpha)}\big\langle\tfrac{k\cdot v_*}{(k\cdot v_*-\alpha-\frac\omega2)(k\cdot(v_*-v)-\omega)}\big\rangle_{v_*}\\
\times\big\langle\tfrac1{k\cdot v_*-\alpha-\frac\omega2}\big\rangle_{v_*}\big\langle\tfrac1{-k\cdot w_*+\alpha-\frac\omega2}\big\rangle_{w_*}d\alpha\bigg),
\end{multline*}
hence, in view of Lemma~\ref{lem:prel-est},
\begin{multline}\label{eq:comput-L3-resolv}
\bigg|\Big\langle\tfrac1{k\cdot (v_*- v)-\omega}(\nabla_{v}-\beta v)\big((i\widehat L_3-i\omega)^{-1}\widehat H^3\big)\big((0,v),(k,v_*),(-k,w_*)\big)\Big\rangle_{v_*,w_*}\bigg|\\
\,\lesssim_\beta\,\langle k\rangle\widehat V(k)|(\nabla_v-\beta v)g^\circ(v)|\Big(\log(2+\tfrac1{\Im\omega})+\int_\R\tfrac1{1+|\alpha+\Re\frac\omega2|}\tfrac1{1+|\alpha-\Re\frac\omega2|}d\alpha\Big)\\
\,\lesssim_\beta\,\log(2+\tfrac1{\Im\omega})\langle k\rangle\widehat V(k)|(\nabla_v-\beta v)g^\circ(v)|,
\end{multline}
where we note that the right-hand side tends to $0$ when multiplied by $\Im\omega$ in the limit $\omega\to0$ with $\Im\omega>0$.
Inserting this bound into~\eqref{eq:g3-dec4} together with the decomposition~\eqref{eq:decomp-L3}, the stated estimate easily follows, and passing to the limit yields in $H^{-1}_\beta(\R^d)$,
\begin{multline*}
\lim_{\omega\to0\atop\Im\omega>0}\int_{\T^d}\big(M_2^1(iL_2-i\omega)^{-1}M_3^2 (iL_3-i\omega)^{-1}S_1^3g^\circ\big)(x,v)\,dx\\
\,=\,(\nabla_{v}-\beta v)\cdot\sum_{k\in2\pi\Z^d}\tfrac{N-1}N\tfrac{\beta\widehat V(k)}{1+\frac{N-1}N\beta\widehat V(k)}\tfrac{i(k\otimes k)\widehat V(k)^2}{\e^\circ_{\beta,1}(k,k\cdot v+i0)}\big\langle\tfrac1{k\cdot (v_*- v)-i0}\big\rangle_{v_*}(\nabla_{v}-\beta v)g^\circ(v).
\end{multline*}
Rewriting $\frac1{\e^\circ_{\beta,1}}$ as in~\eqref{eq:decomp-1/eps}, using the symmetries in $k$, and appealing to the Sokhotski-Plemelj formula as at the end of the proof of Proposition~\ref{prop:2-corr-comp}, the conclusion follows after straightforward simplifications.
\end{proof}

\subsection{Contribution from initial correlations}
We finally show by a direct computation that thanks to the prefactor $i\omega=-\frac{i\alpha+1}{N^r}$ the contribution of initial correlations $NG_N^{3;\circ}$ vanishes in the formula of Corollary~\ref{cor:truncation}. 
\begin{prop}\label{prop:init-corr}
Given $V\in W^{2,\infty}(\T^d)$ and $\beta\in(0,\infty)$, the following convergence holds in $H^{-1}_\beta(\R^d)$, uniformly for $N\ge1$,
\[\lim_{\omega\to0\atop\Im\omega>0}i\omega\int_{\T^d}\big(M_2^1(iL_2-i\omega)^{-1}M_3^2(iL_3-i\omega)^{-1}(NG_N^{3;\circ})\big)(x,v)\,dx\,=\,0,\]
where in addition the argument of the limit can be written as $(\nabla_v-\beta v)G_{N,\beta}^\omega(v)$ where $G_{N,\beta}^\omega(v)$ is bounded pointwise by $C_\beta|(\nabla_v-\beta v)g^\circ(v)|$ uniformly for $\Im\omega>0$.
\end{prop}

\begin{proof}
In view of Lemma~\ref{lem:initial}, we can write $NG_N^{3;\circ}(z_{[0,2]})=g^\circ(v_0) H_{N,\beta}^\circ(x_1-x_2)$ with $\|H_{N,\beta}^\circ\|_{\Ld^2(\T^d)}\lesssim_\beta1$, hence
\begin{equation}\label{eq:decomp-GN30}
N\widehat G_N^{3;\circ}(\hat z_{[0,2]})=\delta(k_1+k_2)\,g^\circ(v_0)\,\widehat H_{N,\beta}^\circ(k_1).
\end{equation}
Combining this with the definitions of $M_2^1$ and $M_3^2$ in Lemma~\ref{lem:eqns}, and repeating the computation of the resolvent of $L_2$ as in~\eqref{eq:L2-inv}, we find
\begin{multline*}
\int_{\T^d}\big(M_2^1(iL_2-i\omega)^{-1}M_3^2 (iL_3-i\omega)^{-1}(N\widehat G_N^{3;\circ})\big)(x,v)\,dx\\
\,=\,(\nabla_{v}-\beta v)\cdot\sum_{k\in2\pi\Z^d}\tfrac{k\widehat V(k)}{\e_{\beta,1}^\circ(k,k\cdot v+\omega)}
\Big\langle \tfrac1{k\cdot(v_*-v)-\omega}\big(\widehat M_3^2 (i\widehat L_3-i\omega)^{-1}(N\widehat G_N^{3;\circ})\big)((-k,v),(k,v_*))\Big\rangle_{v_*},
\end{multline*}
and for $k\ne0$,
\begin{multline*}
\big(\widehat M_3^2 (i\widehat L_3-i\omega)^{-1}(N\widehat G_N^{3;\circ})\big)((-k,v),(k,v_*))\\
\,=\,-\tfrac{N-1}Nik\widehat V(k)
\cdot(\nabla_{v}-\beta v) \big\langle\big((i\widehat L_3-i\omega)^{-1}(N\widehat G_N^{3;\circ})\big)\big((0,v),(k,v_*),(-k,w_*)\big)\big\rangle_{w_*}.
\end{multline*}
Explicitly computing the resolvent of $L_3$ based on Lemmas~\ref{lem:spectrum}(ii) and~\ref{lem:sum-resolv}, and inserting~\eqref{eq:decomp-GN30}, a direct estimate based on Lemma~\ref{lem:prel-est} yields as in~\eqref{eq:comput-L3-resolv},
\begin{multline*}
\bigg|\Big\langle\tfrac1{k\cdot (v_*- v)-\omega}(\nabla_{v}-\beta v)\big((i\widehat L_3-i\omega)^{-1}(N\widehat G_N^{3;\circ})\big)\big((0,v),(k,v_*),(-k,w_*)\big)\Big\rangle_{v_*,w_*}\bigg|\\
\,\lesssim_\beta\,\log(2+\tfrac1{\Im\omega})\langle k\rangle|\widehat H_{N,\beta}^\circ(k)|\,|(\nabla_v-\beta v)g^\circ(v)|.
\end{multline*}
Inserting this into the above, we deduce
\begin{multline*}
\bigg\|\int_{\T^d}\big(M_2^1(iL_2-i\omega)^{-1}M_3^2 (iL_3-i\omega)^{-1}(N\widehat G_N^{3;\circ})\big)(x,\cdot)\,dx\bigg\|_{H^{-1}_\beta(\R^d)}\\
\,\lesssim_\beta\,\log(2+\tfrac1{\Im\omega})\sum_{k\in2\pi\Z^d}\langle k\rangle^3\widehat V(k)^2|\widehat H_{N,\beta}^\circ(k)|\,|(\nabla_v-\beta v)g^\circ(v)|\\
\,\lesssim_\beta\,\log(2+\tfrac1{\Im\omega})|(\nabla_v-\beta v)g^\circ(v)|,
\end{multline*}
and the conclusion follows.
\end{proof}

\subsection{Proof of Theorem~\ref{th:main}}\label{sec:proof-main}
Propositions~\ref{prop:2-corr-comp}, \ref{prop:3-corr-comp}, and~\ref{prop:init-corr} allow to pass to the limit in the different terms in Corollary~\ref{cor:truncation} with $i\omega=-\frac{i\alpha+1}{N^r}$. For $0< r<\frac1{18}$, for all $\phi\in C^\infty_c(\R^+)$, we deduce that the following convergence holds in $H_\beta^{-4}(\R^d)$,
\begin{multline*}
\lim_{N\uparrow\infty}\int_0^\infty\phi(\tau)\,(N\partial_t g_N^{1})|_{t=N^r\tau}\,d\tau\\
=\Big(\int_\R g_\phi\Big)(\nabla_{v}-\beta v)\cdot\bigg(\sum_{k\in2\pi\Z^d}(k\otimes k)\tfrac{\pi\widehat V(k)^2\langle\delta(k\cdot(v_*-v))\rangle_{v_*}}{|\e_{\beta,1}(k,k\cdot v+i0)|^2}\bigg)\nabla_v g^\circ(v).
\end{multline*}
In view of the following identities,
\begin{eqnarray*}
\int_\R g_\phi&=&\int_0^\infty\phi,\\
\e_{\beta,1}^{\circ}(k,k\cdot v-i0)&=&\e_\beta(k,k\cdot v),\\
\langle\delta(k\cdot(v_*-v))\rangle_{v_*}&=&(\tfrac\beta{2\pi|k|^2})^{\frac12} e^{-\frac\beta2(v\cdot\frac k{|k|})^2},
\end{eqnarray*}
this yields the conclusion.\qed

\appendix
\section{Fluctuations around thermal equilibrium}\label{app:lin-setting}

In this appendix, we briefly explain how the techniques developed in this work can be adapted for a rigorous derivation of the linearized Lenard--Balescu equation for fluctuations around the mean-field approximation in a linearized regime at thermal equilibrium.
More precisely, for the system~\eqref{eq:part-dyn}, we start from a global equilibrium for the Liouville equation~\eqref{eq:Liouville}, as given by the Gibbs measure
\[M_{N,\beta}(z_{[N]}):=Z_{N,\beta}^{-1}\,e^{-\frac\beta{2N}\sum_{i\ne j}^NV(x_i-x_j)-\frac\beta2\sum_{i=1}^N|v_i|^2},\]
with normalization factor $Z_{N,\beta}$ and with fixed inverse temperature $\beta\in(0,\infty)$.
We consider chaotic initial data for~\eqref{eq:Liouville} of the form
\begin{equation}\label{eq:CI-pert}
F_{N,\delta}^\circ\,=\,M_{N,\beta}(1+\delta g^\circ)^{\otimes N},
\end{equation}
for some given $g^\circ:\R^d\to\R$, in the tiny perturbation regime
\[\delta\ll\tfrac1N,\]
and we consider its time evolution $F_{N,\delta}$ under the Liouville flow~\eqref{eq:Liouville}. Note that the assumption $\int_{\R^d}g^\circ(v)M_\beta(v)\,dv=0$ ensures that $F_{N,\delta}^\circ$ has mass $1$ for all $\delta$.
In the present perturbative regime $\delta\ll\frac1N$, the initial data $F_{N,\delta}^\circ$ is expanded as follows,
\begin{equation}\label{eq:def-fN0}
F_{N,\delta}^\circ=M_{N,\beta}+\delta H_N^\circ+O(\delta^2N^2)M_{N,\beta},\qquad H_N^\circ(z):=M_{N,\beta}(z)\sum_{i=1}^Ng^\circ(v_i).
\end{equation}
We are thus reduced to analyzing the time evolution $H_N$ of the linearized initial data~$H_N^\circ$ under the Liouville flow~\eqref{eq:Liouville}, and we focus on the velocity distribution of a typical particle,
\begin{equation}\label{eq:def-hN1}
h_N^{1;t}(v)\,:=\,\int_{\T^d}H_N^{1;t}(x,v)\,dx,
\end{equation}
where $H_N^{1}$ denotes the first marginal of $H_N$.
Note that the linearized initial data $H_N^\circ$ is no longer a probability density: in particular there holds $\int_{\D^N}H_N^\circ=0$ and this property is preserved along the flow.
In this context, we expect the time-rescaled linearized velocity density $h_N^{1;Nt}$ to remain close to the solution of the linearization of the Lenard--Balescu equation~\eqref{eq:LB} at the Maxwellian equilibrium $M_\beta$.
Writing the Lenard--Balescu collision operator as
\[\LB(f):=Q(f,f;f),\]
in terms of
\begin{gather*}
Q(f,g;h):=\nabla\cdot\int_{\R^d}B(v,v-v_*;\nabla h)\big(f_*\nabla g-g\nabla_*f_*\big)\,dv_*\\
B(v,v-v_*;\nabla h):=\sum_{k\in\Z^d}(k\otimes k)\,\widehat V(k)^2\tfrac{\delta(k\cdot(v-v_*))}{|\e(k,k\cdot v;\nabla h)|^2}\\
\e(k,k\cdot v;\nabla h):=1+\widehat V(k)\int_{\R^d}\tfrac{k\cdot\nabla h(v_*)}{k\cdot(v-v_*)-i0}\,dv_*,
\end{gather*}
and noting that $Q(M_\beta,M_\beta;h)=0$ for any $h$,
the linearized Lenard--Balescu operator at~$M_\beta$ takes on the following guise (cf.~\cite{Merchant-Liboff-73,Strain-07}),
\[\LLB_\beta h:=Q(h,M_\beta;M_\beta)+Q(M_\beta,h;M_\beta),\]
that is, more explicitly, after some simplifications,
\begin{equation*}
\LLB_\beta h:=\nabla\cdot\Big(\int_{\R^d}B(v,v-v_*;\nabla M_\beta)\,\big((M_\beta)_*(\nabla+\beta v)h-M_\beta(\nabla_*+\beta v_*)h_*\big)\,dv_*\Big).
\end{equation*}
Note that the Fokker-Planck operator in~\eqref{eq:FP} coincides with this full linearized Lenard--Balescu operator without loss term.
Proceeding to a similar cumulant analysis as for Theorem~\ref{th:main}, we are led to the following.

\begin{theor}\label{th:main2}
Let the same assumptions hold as in Theorem~\ref{th:main} with $h^\circ:=M_\beta g^\circ$.
Then, for $0<r<\frac1{18}$, the velocity distribution $h_N^1$ of a typical particle (cf.~\eqref{eq:def-hN1}) satisfies on the timescale $t\sim N^r$, in the above linearized regime,
\[\lim_{N\uparrow\infty}N(\partial_th_N^1)|_{t=N^r\tau}\,=\,\LLB_\beta h^\circ,\]
as a function of $(\tau,v)$ in the weak sense of $\Dc'(\R^+\times\R^d)$.
\end{theor}

\begin{proof}[Idea of the proof]
We only indicate the main differences with the proof of Theorem~\ref{th:main} and we omit the detail. We split the proof into three main steps.

\medskip
\step1 Cumulant expansion.\\
For all $1\le m\le N$, denoting by $H_N^m$ the $m$th-order marginal of $H_N$,
its $m$th-order cumulant is naturally defined as follows,
\[G_N^m(z_{[m]})\,=\,\sum_{j=1}^m(-1)^{m-j}\sum_{\sigma\in\mathfrak P^m_j}\tfrac{H_N^j}{M_\beta^{\otimes j}}(z_\sigma),\]
which indeed corresponds to the definition in Lemma~\ref{lem:cumulants} once the tagged particle is removed. The a priori estimates of Lemma~\ref{lem:cum-ap} then take the following form, for $1\le m\le N$,
\[\|G_N^{m;t}\|_{\Ld^2_\beta(\D^m)}\,\lesssim_{\beta,m}\,N^{-\frac{m-1}2}.\]
Next, similarly as in Lemma~\ref{lem:eqns}, the cumulant $G_N^m$ satisfies
\[\partial_tG_N^{m}+iL_{m}G_N^{m}\,=\,M_{m+1}^{m}G_N^{m+1}+\tfrac1N\big(S_{m-2}^{m}G_N^{m-2}+S_{m-1}^{m}G_N^{m-1}+S_{m}^{m}G_N^{m}\big),\]
where the operators $L_m$, $M_{m+1}^m$, $S_{m-2}^m$, $S_{m-1}^m$, and $S_m^m$ act similarly as the corresponding operators $L_{m+1}$, $M_{m+2}^{m+1}$, $S_{m-1}^{m+1}$, $S_{m}^{m+1}$, and $S_{m+1}^{m+1}$ in Lemma~\ref{lem:eqns} when restricted to the last $m$ exchangeable variables.
In particular, we now define $L_m=\sum_{j=1}^mL_{m+1}^{(j)}$ with $L_{m+1}^{(j)}$ as in Lemma~\ref{lem:eqns}.
This new hierarchy of equations can be truncated on intermediate timescale similarly as in Section~\ref{sec:intermediate}, and it remains to compute its Markovian limit.

\medskip
\step2 Contribution from $2$-particle correlations.\\
The contribution of $2$-particle correlations to $N(\partial_tg_N^1)|_{t=N^r\tau}$ as $N\uparrow\infty$ with $0<r<\frac1{18}$ takes the form $M_2^1(iL_2+0)^{-1}S_1^2g^\circ$ with the new notation for the operators $M_2^1$, $L_2$, and~$S_1^2$, and we prove in $H^{-1}_\beta(\R^d)$ that
\begin{multline*}
\lim_{N\uparrow\infty,\,\omega\to0\atop\Im\omega>0}\int_{\T^d}\big(M_2^1(iL_2-i\omega)^{-1}S_1^2g^\circ\big)(x,v)\,dx\\
\,=\,(\nabla_{v}-\beta v)\cdot\sum_{k\in2\pi\Z^d}(k\otimes k)\tfrac{\pi\widehat V(k)^2}{|\e^\circ_{\beta,2}(k,k\cdot v+i0)|^2}
\bigg(\Big\langle \delta\big(k\cdot (v_*-v)\big)\,\big(\nabla_vg^\circ(v)-\nabla_{v_*}g^\circ(v_*)\big)\Big\rangle_{v_*}\\
+\beta\widehat V(k)\Big\langle\delta\big(k\cdot (v_*-v)\big) \big((\nabla_v-\beta v)g^\circ(v)-(\nabla_{v_*}-\beta v_*)g^\circ(v_*)\big)\Big\rangle_{v_*}\bigg).
\end{multline*}
By definition of $M_2^1$, we can write in Fourier variables, as in~\eqref{eq:comput-M21-g2},
\begin{multline*}
\int_{\T^d}\big(M_2^1(iL_2-i\omega)^{-1}S_1^2g^\circ\big)(x,v)\,dx\\
\,=\,-\tfrac1{2i}\sum_{k\in2\pi\Z^d}k\widehat V(k)\cdot(\nabla_{v}-\beta v)
\Big\langle\big((i\widehat L_2-i\omega)^{-1}\widehat S_1^2g^\circ\big)((-k,v),(k,v_*))\\
-\big((i\widehat L_2-i\omega)^{-1}\widehat S_1^2g^\circ\big)((k,v),(-k,v_*))\Big\rangle_{v_*},
\end{multline*}
while the resolvent of $L_2$ now takes on the following guise,
\begin{align}\label{eq:pre-decomp-LB}
&\big\langle\big((i\widehat L_2-i\omega)^{-1}\widehat S_1^2g^\circ\big)((-k,v),(k,v_*))\big\rangle_{v_*}
\,=\,\tfrac{-i}{\e_{\beta,2}^\circ(k,k\cdot v+\omega)}\Big\langle\tfrac{\widehat S_1^2g^\circ((-k,v),(k,v_*))}{k\cdot (v_*-v)-\omega}\Big\rangle_{v_*}\\
&-\tfrac{N-1}N\beta\widehat V(k)\tfrac1{2\pi}\int_\R\tfrac1{\e_{\beta,2}^\circ(k,\frac\omega2-\alpha)\e_{\beta,2}^\circ(-k,\alpha+\frac\omega2)}\tfrac{k\cdot v}{k\cdot v+\alpha+\frac\omega2}\Big\langle\tfrac{\widehat S_1^2g^\circ((-k,v),(k,v_*))}{(k\cdot v+\alpha+\frac\omega2)(k\cdot v_*+\alpha-\frac\omega2)}\Big\rangle_{v,v_*}\,d\alpha,\nonumber
\end{align}
which is obtained from Lemmas~\ref{lem:spectrum}(ii) and~\ref{lem:sum-resolv} together with a complex deformation argument as in~\eqref{eq:L2-inv}. The first right-hand side term has the same structure as in~\eqref{eq:L2-inv}, while the second one is new and cannot be computed by complex deformation due to the presence of poles in both complex half-planes. Formally passing to the limit and noting that $\e^\circ_{\beta,2}(-k,\alpha+i0)=\overline{\e^\circ_{\beta,2}(k,-\alpha+i0)}$ and $\e^\circ_{\beta,2}(k,\omega)\to\e^\circ_{\beta,1}(k,\omega)$, we are led to
\begin{multline}\label{eq:comput-lim0}
\lim_{N\uparrow\infty,\,\omega\to0\atop\Im\omega>0}\big\langle\big((i\widehat L_2-i\omega)^{-1}\widehat S_1^2g^\circ\big)((-k,v),(k,v_*))\big\rangle_{v_*}
\,=\,\tfrac{-i}{\e^\circ_{\beta,1}(k,k\cdot v+i0)}\Big\langle\tfrac{\widehat S_1^2g^\circ((-k,v),(k,v_*))}{k\cdot (v_*-v)-i0}\Big\rangle_{v_*}\\
-\tfrac1{2\pi}\int_\R\tfrac{\beta\widehat V(k)}{|\e^\circ_{\beta,1}(k,-\alpha+i0)|^2}\tfrac{k\cdot v}{k\cdot v+\alpha+i0}\Big\langle\tfrac{\widehat S_1^2g^\circ((-k,v),(k,v_*))}{(k\cdot v+\alpha+i0)(k\cdot v_*+\alpha-i0)}\Big\rangle_{v,v_*}\,d\alpha,
\end{multline}
The definition of $S_{1}^2$ yields 
\[\widehat S_{1}^2g^\circ((-k,v),(k,v_*))=-i\widehat V(k)\big(S_k(v,v_*)-S_k(v_*,v)\big)\]
with $S_k(v,v_*):=k\cdot(\nabla_{v}-\beta v+\beta v_*)g^\circ(v)$. The main difference with the case of the tagged particle is that now $v$ and $v_*$ play a symmetric role. In particular,  we deduce from this symmetry that
\[\Big\langle\tfrac{\widehat S_1^2g^\circ((-k,v),(k,v_*))}{(k\cdot v+\alpha+i0)(k\cdot v_*+\alpha-i0)}\Big\rangle_{v,v_*}\,=\,2\widehat V(k)\,\Im\Big\langle\tfrac{S_k(v,v_*)}{(k\cdot v+\alpha+i0)(k\cdot v_*+\alpha-i0)}\Big\rangle_{v,v_*}\]
is purely real.
Also noting that exchanging $k$ and $-k$ in the right-hand side of~\eqref{eq:comput-lim0} amounts to taking the complex conjugate, the above computations yield
\begin{eqnarray*}
\lefteqn{\hspace{-1.6cm}\lim_{N\uparrow\infty,\,\omega\to0\atop\Im\omega>0}\tfrac1{2i}\Big\langle\big((i\widehat L_2-i\omega)^{-1}\widehat S_1^2g^\circ\big)((-k,v),(k,v_*))-\big((i\widehat L_2-i\omega)^{-1}\widehat S_1^2g^\circ\big)((k,v),(-k,v_*))\Big\rangle_{v_*}}\\
&=&\lim_{N\uparrow\infty,\,\omega\to0\atop\Im\omega>0}\Im\big\langle\big((i\widehat L_2-i\omega)^{-1}\widehat S_1^2g^\circ\big)((-k,v),(k,v_*))\big\rangle_{v_*}\\
&=&-\Im\Big(\tfrac{\widehat V(k)}{\e^\circ_{\beta,1}(k,k\cdot v+i0)}\Big\langle\tfrac{S_k(v,v_*)-S_k(v_*,v)}{k\cdot (v_*-v)-i0}\Big\rangle_{v_*}\Big)\\
&&-\tfrac1{\pi}\int_\R\tfrac{\beta\widehat V(k)^2}{|\e^\circ_{\beta,1}(k,-\alpha+i0)|^2}\Im\big(\tfrac{k\cdot v}{k\cdot v+\alpha+i0}\big)\Im\Big\langle\tfrac{S_k(v,v_*)}{(k\cdot v+\alpha+i0)(k\cdot v_*+\alpha-i0)}\Big\rangle_{v,v_*}\,d\alpha.
\end{eqnarray*}
Now appealing to the Sokhotski-Plemelj identity in the form $\Im\tfrac{k\cdot v}{k\cdot v+\alpha+i0}=-\pi\delta(k\cdot v+\alpha)$, this becomes
\begin{multline*}
\lim_{N\uparrow\infty,\,\omega\to0\atop\Im\omega>0}\tfrac1{2i}\Big\langle\big((i\widehat L_2-i\omega)^{-1}\widehat S_1^2g^\circ\big)((-k,v),(k,v_*))-\big((i\widehat L_2-i\omega)^{-1}\widehat S_1^2g^\circ\big)((k,v),(-k,v_*))\Big\rangle_{v_*}\\
\,=\,-\Im\Big(\tfrac{\widehat V(k)}{\e^\circ_{\beta,1}(k,k\cdot v+i0)}\Big\langle\tfrac{S_k(v,v_*)-S_k(v_*,v)}{k\cdot (v_*-v)-i0}\Big\rangle_{v_*}\Big)
+\tfrac{\beta\widehat V(k)^2(k\cdot v)}{|\e^\circ_{\beta,1}(k,k\cdot v+i0)|^2}\,\Im\Big\langle\tfrac{S_k(v_*,w_*)}{(k\cdot (v_*-v)+i0)(k\cdot (w_*-v)-i0)}\Big\rangle_{v_*,w_*}.
\end{multline*}
Expanding
\begin{eqnarray*}
\tfrac1{\e^\circ_{\beta,1}(k,k\cdot v+i0)}\,=\,\tfrac1{|\e^\circ_{\beta,1}(k,k\cdot v+i0)|^2}\big(1+\beta\widehat V(k)+\beta\widehat V(k)(k\cdot v)\big\langle\tfrac{1}{k\cdot(v_*-v)+i0}\big\rangle_{v_*}\big),
\end{eqnarray*}
and appealing to the definition of $S_k$ and to the choice $\langle g^\circ(v)\rangle_v=0$ in form of
\begin{eqnarray*}
\Big\langle\tfrac{S_k(v,v_*)-S_k(v_*,v)}{k\cdot (v_*-v)-i0}\Big\rangle_{v_*}&=&k\cdot\Big\langle\tfrac{\nabla_vg^\circ(v)-\nabla_{v_*}g^\circ(v_*)}{k\cdot (v_*-v)-i0}\Big\rangle_{v_*}+\beta g^\circ(v),\\
\Big\langle\tfrac{S_k(v_*,w_*)}{(k\cdot (v_*-v)+i0)(k\cdot (w_*-v)-i0)}\Big\rangle_{v_*,w_*}&=&k\cdot\Big\langle\tfrac{\nabla_{v_*}g^\circ(v_*)}{k\cdot (v_*-v)+i0}\Big\rangle_{v_*}\big\langle\tfrac1{k\cdot (v_*-v)-i0}\big\rangle_{v_*}+\beta\Big\langle\tfrac{g^\circ(v_*)}{k\cdot(v_*-v)+i0}\Big\rangle_{v_*},
\end{eqnarray*}
we find after straightforward simplifications,
\begin{multline*}
\lim_{N\uparrow\infty,\,\omega\to0\atop\Im\omega>0}\tfrac1{2i}\Big\langle\big((i\widehat L_2-i\omega)^{-1}\widehat S_1^2g^\circ\big)((-k,v),(k,v_*))-\big((i\widehat L_2-i\omega)^{-1}\widehat S_1^2g^\circ\big)((k,v),(-k,v_*))\Big\rangle_{v_*}\\
\,=\,-\tfrac{\widehat V(k)(1+\beta\widehat V(k))}{|\e^\circ_{\beta,1}(k,k\cdot v+i0)|^2}\Im\Big\langle\tfrac{k\cdot(\nabla_vg^\circ(v)-\nabla_{v_*}g^\circ(v_*))}{k\cdot (v_*-v)-i0}\Big\rangle_{v_*}
-\tfrac{\beta^2\widehat V(k)^2(k\cdot v)}{|\e^\circ_{\beta,1}(k,k\cdot v+i0)|^2}\Im\Big\langle\tfrac{g^\circ(v)-g^\circ(v_*)}{k\cdot(v_*-v)+i0}\Big\rangle_{v_*},
\end{multline*}
and the claim follows from the Sokhotski-Plemelj identity.

\medskip
\step3 Contribution from $3$-particle correlations.\\
The contribution of $3$-particle correlations takes the form $M_2^1(iL_2+0)^{-1}M_3^2(iL_3+0)^{-1}S_1^3g^\circ$, and we prove in $H^{-1}_\beta(\R^d)$ that
\begin{multline*}
\lim_{N\uparrow\infty,\,\omega\to0\atop\Im\omega>0}\int_{\T^d}\big(M_2^1(iL_2+0)^{-1}M_3^2(iL_3+0)^{-1}S_1^3g^\circ\big)(x,v)\,dx\\
\,=\,-(\nabla_v-\beta v)\cdot\sum_{k\in2\pi\Z^d}\beta\widehat V(k)(k\otimes k)\tfrac{\pi\widehat V(k)^2}{|\e^\circ_{\beta,1}(k,k\cdot v+i0)|^2}\\
\times\Big\langle\delta\big(k\cdot(v_*-v)\big)\big((\nabla_v-\beta v)g^\circ(v)-(\nabla_{v_*}-\beta v_*)g^\circ(v_*)\big)\Big\rangle_{v_*}.
\end{multline*}
The explicit computation of $(iL_3+0)^{-1}S_1^3g^\circ$ is performed similarly as in the proof of Proposition~\ref{prop:3-corr-comp}, while the resolvent of $L_2$ is computed as in Step~2, and the claim easily follows.
\qedhere
\end{proof}

\section*{Acknowledgements}
We wish to thank Pierre-Emmanuel Jabin and Freddy Bouchet for inspiring discussions.
The work of MD is supported by the CNRS-Momentum program.

\bibliographystyle{plain}
\bibliography{biblio}

\end{document}